\documentclass[review]{elsarticle}

\usepackage{hyperref}

\usepackage{amsfonts,amsmath,amsthm,bm} 
\usepackage{algpseudocode}
\usepackage{algorithm}
\usepackage{multirow}
\usepackage{slashbox}
\usepackage{subfigure}
\usepackage{graphicx}
\usepackage{tikz}
\usepackage{slashbox}
\usepackage{lscape}
\usepackage{cancel}

\allowdisplaybreaks

\journal{Computers \& Mathematics with Applications}

\newcommand{\R}{\mathbb{R}}

\newcommand{\Q}{\mathcal{Q}}

\newcommand{\E}{\mathbb{E}}

\newcommand{\rd}{{\rm d}}

\theoremstyle{definition}
\newtheorem{definition}{Definition}[section]

\theoremstyle{corollary}

\theoremstyle{remark}

\theoremstyle{proposition}
\newtheorem{proposition}{Proposition}[section]

\begin{document}

\begin{frontmatter}

\title{PDEs for pricing interest rate derivatives under the new generalized Forward Market Model (FMM)}

%
%
%

\author[auth1]{J. G. L\'opez-Salas}\ead{jose.lsalas@udc.es}
\address[auth1]{Department of Mathematics, Faculty of Informatics and CITIC, Campus Elvi\~na s/n, 15071-A Coru\~na (Spain)}

\author[auth2]{S. P\'erez-Rodr\'iguez}\ead{sperezr@ull.edu.es}
\address[auth2]{Departamento de An\'alisis Matem\'atico and IMAULL, Universidad de La Laguna, 38208, La Laguna, Tenerife, Canary Islands (Spain)}

\author[auth1]{C. V\'azquez}\ead{carlosv@udc.es}

\begin{abstract}
In this article we derive partial differential equations (PDEs) for pricing interest rate derivatives under the generalized Forward Market Model (FMM) recently pre\-sen\-ted by A. Lyashenko and F. Mercurio in \cite{lyashenkoMercurio:Mar2019} to model the dynamics of the Risk Free Rates (RFRs) that are replacing the traditional IBOR rates in the financial industry. Moreover, for the numerical solution of the proposed PDEs formulation, we develop some adaptations of the finite differences methods developed in \cite{LopezPerezVazquez:sisc} that are very suitable to treat the presence of spatial mixed derivatives. This work is the first article in the literature where PDE methods are used to value RFR derivatives. Additionally, Monte Carlo-based methods will be designed and the results are compared with those obtained by the numerical solution of PDEs.
\end{abstract}

\begin{keyword}
IBOR replacement, generalized forward market model, forward rates, finite differences, AMFR-W methods.
\MSC[2020] 91G30, 91G60, 65M06, 65M22
\end{keyword}

\end{frontmatter}


\section{Introduction}

For decades, financial institutions have been using InterBank Offered Rates (IBORs) as reference rates to determine interest or as underlyings of derivatives in several currencies, perhaps the more popular example being the London InterBank Offered Rate (LIBOR). More than 350 trillion dollars in derivatives and other financial products are tied to these rates. 

IBORs reflect the average unsecured short-term interest rate at which large global banks can borrow from each other. They are based on surveys to banks and not on real transactions.

At the beginning of the 21st century, several big banks manipulated the interest rate they reported that they could borrow at. Firstly, to allow their traders, who were taking derivative bets on where this IBOR would be, to make more money because the rate was artificially distorted. Later, in the depth of the 2008 financial crisis, banks again manipulated IBORs, this time not to make profits from derivatives trading, but to make themselves look financially stronger than they were \cite{FabriziHuanParbonetti:2021,HuanPrevitsParbonetti:2022,McConnell,Gandhi}.

In view of previous IBORs scandals, a few years ago financial authorities worldwide initiated the replacement of IBORs with alternative Risk Free Rates (RFRs). RFRs are reported to be robust because they rely on real transactions. The whole banking industry is adapting its products by offering them based on RFR for new trades (see \cite{Duffie:2018, HuangTodorov:2022}, for example). This transition, known as IBOR transition, is complex for clients, dealers, and financial authorities. 

In fact, on December 31st, 2021, some IBORs ceased to be published and became non-representative. This concerned all tenors of IBOR Japanese Yen, British Pounds, Swiss Francs, Euros, and two tenors of IBOR USD dollars. Recently, on June 30th, 2023, the remaining tenors of IBOR USD dollars also stopped being reported and became irrelevant. For the time being, all major economies have selected RFRs to replace their corresponding IBORs. For example, the United States of America adopted SOFR (Secured Overnight Financing Rate), the European Union selected ESTER (Euro Short-TErm Rate), the United Kingdom designed SONIA (Sterling Over\-Night Index Average), Switzerland took SARON (Swiss Average Rate OverNight), and Japan selected TONAR or TONA (Tokyo OverNight Average Rate).

The main general features of RFRs are the following. Firstly, RFRs are overnight rates and not term rates like IBORs (i.e. one week, one month, three months, ...). Secondly, by definition, RFRs are \textit{backward-looking}, which means that the rate to be paid for the application period is calculated by reference to historical transaction data and set at the end of that time interval. Unlike RFRs, IBORs are \textit{forward-looking} rates, meaning that the rate to be paid for the application period is set at the beginning of that time interval. Additionally, RFRs are risk-free since one-day credit risk can be neglected. On top of that, RFRs not only represent the interbank market; in fact they are rates for the entire market.

The LIBOR Market Model (LMM) was widely used by financial ins\-ti\-tu\-tions for the valuation of interest rate derivatives based on IBORs (see the seminal works \cite{BraceGatarekMusiela,Jamshidian} and the book \cite{brigoMercurio}, for example). The main reason behind its great success comes from the fact that LMM considers rate dynamics consistent with the well-established Black-Scholes market formulas for pricing caplets and swaptions. Since then, a lot of work has been devoted to the pricing of interest rate derivatives by using different methods, such as Monte Carlo or PDEs. Particularly, concerning PDEs formulations, the pricing of some complex derivatives with classical LMM has been addressed in \cite{Pietersz, PascucciSuarezVazquez:mmmas, SuarezVazquez:mcm, SuarezVazquez:amc}, among others. Extensions of the LMM framework to incorporate stochastic volatility have been discussed in \cite{MercurioMorini, RebonatoWhite, annurevRebonato,alanBrace} and references therein, with a first formulation in terms of PDEs developed in \cite{LopezVazquez:cmwa}. Thus, in the setting of LMM, the formulation in terms of PDEs for pricing different interest rate derivatives has been widely studied. It is important to notice that the LMM contemplates only forward-looking rates. Therefore, it is no longer valid to price financial products based on the new RFRs, that are backward-looking.

Nowadays, the community working on quantitative finance is very active in proposing new mathematical models able to price the new derivatives based on RFRs. Having in mind that RFRs can be converted into compounded setting-in-arrears term rates, mathematical models for pricing RFR-based derivatives can mainly follow two different approaches. The first strategy is to directly simulate daily the underlying RFRs in their corresponding application periods. The second approach models term rates based on RFRs. The most promising interest rate model following the second strategy is the one proposed by Lyashenko and Mercurio in \cite{lyashenkoMercurio:Mar2019}. This model, referred to as the generalized Forward Market Model (FMM), is the main starting point in this work.

The FMM is a modeling framework that allows for the joint modeling of forward-looking and backward-looking rates inside the same parsimonious setup. Actually, FMM is an extension of the successful LMM. More precisely, in the post-LIBOR world, IBORs have to be replaced with some more general forward rates, and that is exactly what explains the term \textit{generalized forward} market model that has been coined for this new setting.

FMM accommodates both forward-looking and backward-looking rates inside the same framework in a very natural way. Indeed, the FMM incorporates additional advantages over the LMM. One of them comes from the possibility of modeling rates directly under the classic risk-neutral money-market measure, something that was not possible with the LMM. In \cite{lyashenkoMercurio:Nov2019}, Lyashenko and Mercurio explain how to complete the FMM in order to generate rates that are outside the given time grids that are initially assigned to the model. The authors build a general Heath–Jarrow–Morton (HJM) model that originates the FMM. Once this HJM model is constructed, since it is very general, one can create general rates and discount factors. This approach was not possible in the classical LMM. In fact, under LMM the approach is usually to complete the model by adding some
interpolation method, which is typically called rate interpolation.

In order to price RFR-based derivatives under the FMM when the payoff depends on the joint distribution of several interest rates, numerical methods must be considered. For this purpose, expectation-based formulations or PDE-based formulations of the pricing problem can be mainly used. Although the standard approach uses Monte Carlo simulations for expectation-based formulations, it exhibits several drawbacks when pricing interest rate derivatives, as it has been pointed out in \cite{LopezVazquez:cmwa}, for example. Although some disadvantages could be smoothed or even removed by enhanced Monte Carlo techniques applied to specific interest rate derivatives, the pricing of early exercise derivatives based on RFRs, like Bermudan swaptions, would require a highly computational demanding suitable adaption of Monte Carlo methodology. In the PDEs formulation setting, pricing these interest rate derivatives does not involve a huge increase of computational time with respect to analogous derivatives without early exercise opportunity.


These previous arguments motivate the interest in developing suitable PDE formulations for solving the pricing problem in the new recently established FMM model. In this work, we formulate the pricing of RFRs derivatives under the FMM in terms of PDEs. To the best of our knowledge, our presentation is the first in the literature.

Moreover, we present an efficient and recent numerical algorithm to approximate the solution of the proposed PDEs formulation by using finite differences and the AMFR-W1 method for the time integration, as already used in \cite{LopezPerezVazquez:sisc}. This method belongs to the class of AMFR-W-methods \cite{GlezHairerHdezPerez18,GlezHdezPerez21}, which are very efficient when dealing with parabolic problems involving mixed derivatives, as they avoid computing explicitly the part of the Jacobian that includes the discretization of such mixed derivatives. In \cite{LopezPerezVazquez:sisc}, a numerical strategy that combines appropriate finite differences schemes to deal with terms containing mixed derivatives with sparse grids technique has been successfully used for pricing interest rate derivatives when classical LMM for forward rates including stochastic volatility is considered. 

However, in the present work, our aim will be different. More precisely, as the payoff function of the derivative that determines the dynamics of the PDE has differentiability issues near the strike values, we have explored the integration on non-uniform meshes, which contain many more points near the payoff non-differentiability area than in the rest of the domain. As we will see, the consideration of appropriate non-uniform meshes improves the accuracy and reliability of the approximation and we will obtain an approximation of higher quality than that provided by the Monte Carlo method, at least when the ``spatial''\footnote{When working with time-dependent PDEs such as the ones we consider in this work, it is common to use physics-like terminology where the word ``spatial'' refers to variables other than time. In our case, the "spatial" variables are the forward rates, so that the number of these rates is the ``spatial dimension'' of the PDE. In what follows, we will use this terminology.} dimension of the PDE is less than or equal to 4. The application of sparse grids to solve the so-called curse of dimensionality in these PDEs will be left for future work.

The article is organized as follows. In Section \ref{sec:definitions} we review the extended zero-coupon bonds, which are the cornerstone of the recent FMM. A detailed description of how this concept of extended bonds allows to define not only the classical forward-looking forward rates but also the novel backward-looking forward rates, is also carried out. Besides, a thorough illustration of how to compute the extended discount factors from forward rates values is presented. In Section \ref{sec:FMM} the system of stochastic differential equations (SDEs) of the FMM is introduced, considering joint dynamics for the interest rates. In Section \ref{sec:PDEsFMM}, we derive the  PDEs for the FMM. Next, numerical methods to price derivatives under the FMM are designed in Section \ref{sec:numericalMethods}. In Section \ref{sec:numericalExperiments}, numerical experiments are carried out to assess the behavior of the developed numerical methods. Finally, conclusions and future work are discussed in Section \ref{sec:conclusions}.

\section{Main assumptions, definitions and notations} \label{sec:definitions}

In this section, we present the basic notations and definitions that will be used throughout the article. A continuous-time financial market is considered. It has an instantaneous RFR whose value at time $t$ is denoted by $r(t)$. 

\begin{definition}[Bank account]
Let $B(t)$ be the value of the bank account at time $t\geq 0$. $B$ is the classic process that satisfies the ordinary differential equation $\rd B(t) = r(t) B(t)\, \rd t$ with $B(0)=1$, so that $B(t)=e^{\int_0^t r(u) \rd u}$.
\end{definition}

Moreover, we assume the existence of a risk-neutral measure $\Q$, whose associated numeraire is the bank account $B$ we have just defined. Besides, $\E$ will denote the expectation with respect to the risk-neutral measure, and $\mathcal{F}_t$ will be the $\sigma$-algebra generated by risk factors up to the evaluation time.

\begin{definition}[Zero-coupon bond price]
A zero-coupon bond with maturity $T$ is a very simple contract that pays its holder one unit of currency at time $T$, with no intermediate payments. For $t<T$, let $P(t,T)$ be the value at time $t$ of this product. We have the following valuation formula, which is given by risk-neutral pricing:
\begin{equation} \label{eq:zeroCouponBondPrice}
 P(t,T)=\E\left[ e^{-\int_t^T r(u) \rd u }\Big| \mathcal{F}_t\right].
\end{equation}
Note that $P(T,T)=1$ for all $T$.
\end{definition}

The formula \eqref{eq:zeroCouponBondPrice} is clearly defined for valuation times $t$ before the maturity $T$ ($t\leq T$) since one wants to calculate the bond price before it expires. In the new FMM framework, the advantage is that it is mathematically possible to define $P(t,T)$ even for those times $t$ after the maturity $T$ ($t>T$).

\begin{definition}[Extended zero-coupon bond price] \label{def:extendedBondPrice}
For $t>T$, Equation \eqref{eq:zeroCouponBondPrice} reduces to
\begin{equation} \label{eq:extendedZeroCouponBondPrice}
P(t,T)=\E\left[ e^{\int_T^t r(u) \rd u }\Big| \mathcal{F}_t\right] = e^{\int_T^t r(u) \rd u } = \dfrac{B(t)}{B(T)}.
\end{equation}
Note that $P(t,0)=B(t)$.
\end{definition}

One must observe that the integral in equation \eqref{eq:zeroCouponBondPrice} can be also defined for $t>T$, and it is equal to $e^{\int_T^t r(u) \rd u}$, which is measurable with respect to $\mathcal{F}_t$. Then, using properties of conditional expectations, the conditioned expected value is equal to the term inside (see \cite{Mikosh98}, for example). Finally, one just uses the definition of $B(t)$ to obtain that the price of the bond after its maturity is equal to the bank account at the valuation time $t$ divided by the bank account at the maturity $T$.

This concept of extended zero coupon bonds is taken from \cite{lyashenkoMercurio:Mar2019}. Although it was already considered in \cite{AndersenPiterbarg} and \cite{GlassermanZhao00} to define hybrid numeraires and measures, in \cite{lyashenkoMercurio:Mar2019} it is used for the first time to extend the definition from forward-looking to backward-looking rates, so that also forward rates can be appropriately extended after their maturity dates.

\begin{definition}[Extended $T$-forward measure]
 The extended $T$-forward measure, denoted by $\Q^T$, is the martingale measure associated with the extended bond price $P(t,T)$. Note that the risk-neutral measure is a particular case of the extended $T$-forward measure where $T=0$, i.e $\Q = \Q^0$.
\end{definition}


\subsection{The compounded setting-in-arrears term rate}

Financial derivatives written on RFRs consider as underlyings daily compounded setting-in-arrears term rates, which by definition are backward-looking in nature. Hereafter, $N\geq 1$ denotes the number of rates to be modeled. Let us define them in the next paragraphs.

We start with the tenor structure $0=T_0<T_1<\ldots<T_N$. Let $\tau_k$ be the year fraction of the $k$-th time interval $[T_{k-1},T_k)$. Next, we define the backward-looking spot rate.

\begin{definition}[Backward looking spot rate]
The simple backward-looking spot rate is defined as
$$R(T_{k-1},T_k) = \dfrac{1}{\tau_k} \left[ e^{\int_{T_{k-1}}^{T_k} r(u) \rd u} -1 \right] = \dfrac{1}{\tau_k} \left[ \dfrac{B(T_k)}{B(T_{k-1})} -1 \right] = \dfrac{1}{\tau_k} \left[ P(T_k,T_{k-1}) -1 \right].$$
$R(T_{k-1},T_k)$ is the simple interest rate such that the investment of one unit of currency at time $T_{k-1}$ yields $P(T_k,T_{k-1})$ units of currency at time $T_{k}$.
\end{definition}

Additionally, we also need to consider forward-looking rates, which are the same as in the LMM (see \cite{brigoMercurio}).

\begin{definition}[Forward-looking spot rate]
The simple forward-looking spot rate is defined as
$$F(T_{k-1},T_k) = \dfrac{1}{\tau_k} \left[ \dfrac{1}{P(T_{k-1},T_{k})} -1 \right].$$
$F(T_{k-1},T_k)$ is the simple interest rate such that the investment of $P(T_{k-1},T_{k})$ units of currency at time $T_{k-1}$ yields one unit of currency at time $T_{k}$.
\end{definition}

\subsection{Forward rates}

Once we have defined spot rates, we move to the definition of forward rates. 

\begin{definition}[Backward-looking forward rate]
The simple compounded back\-ward-looking forward rate prevailing at time $t$ for the time interval $[T_{k-1},T_k)$ is denoted by $R_k(t)$ and defined by 
\begin{equation} \label{eq:forwardOfBackward}
 R_k(t)=\dfrac{1}{\tau_k} \left(\dfrac{P(t,T_{k-1})}{P(t,T_{k})} -1\right).
\end{equation}
It is the value of the fixed rate $K_R$ in the swaplet paying $\tau_k (R(T_{k-1},T_k) - K_R)$ at time $T_k$, such that this product has zero value at time $t$ (see Figure \ref{fig:FRA_forwardBackward}). Note that the definition \eqref{eq:forwardOfBackward} is valid for all times $t$, even those times $t>T_k$. The rate
$R_k(t)$ satisfies the following properties:
\begin{itemize}
 \item $R_k(T_{k-1})=F(T_{k-1},T_k)$, i.e., at time $T_{k-1}$ it is equal to the forward-looking spot rate.
 \item $R_k(T_k)=R(T_{k-1},T_k)$, i.e., at time $T_k$ it is equal to the backward-looking spot rate.
 \item For $t>T_k$, $R_k(t)=R(T_{k-1},T_k)$, i.e., after time $T_k$ it stops evolving.
\end{itemize}

\end{definition}

\begin{figure}[!h]
\begin{center}
\begin{tikzpicture}
\draw[] (0,0) -- (7,0);
 \draw (1,0) -- (1,3pt);
 \draw (1,0) -- (1,-3pt);
 \draw (1,-0.1) node[below] {$t$};

 \draw (4,0) -- (4,3pt);
 \draw (4,0) -- (4,-3pt);
 \draw (4,-0.1) node[below] {$T_{k-1}$};
 
  \draw (6,0) -- (6,3pt);
 \draw (6,0) -- (6,-3pt);
 \draw (6,-0.1) node[below] {$T_{k}$};
 
 \draw (6,0) node[above] {$\uparrow$};
  \draw (6,0.4) node[above] {$\tau_k(R(T_{k-1},T_{k})-K_R)$};
\end{tikzpicture}
\end{center}
\caption{Swaplet based on the backward-looking rate.}
\label{fig:FRA_forwardBackward}
\end{figure}
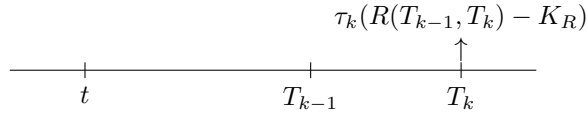

\begin{definition}[Forward-looking forward rate]
The simple compounded forward-looking forward rate prevailing at time $t$ for the time interval $[T_{k-1},T_k)$ is denoted by $F_k(t)$ and defined by 
\begin{equation} \label{eq:forwardOfForward}
             F_k(t) = \left\{ \begin{array}{lcc}
             R_k(t) &   \mbox{if}  & t \leq T_{k-1} \\
             F(T_{k-1},T_k) &  \mbox{if} & t > T_{k-1}.
             \end{array}
   \right.
\end{equation}
It is the value of the fixed rate $K_F$ in the swaplet paying $\tau_k (R(T_{k-1},T_k) - K_F)$ at time $T_k$ such that this product has zero value at time $t$, see Figure \ref{fig:FRA_forwardForward}.
\end{definition}

\begin{figure}[!h]
\begin{center}
\begin{tikzpicture}
\draw[] (0,0) -- (7,0);
 \draw (1,0) -- (1,3pt);
 \draw (1,0) -- (1,-3pt);
 \draw (1,-0.1) node[below] {$t$};

 \draw (4,0) -- (4,3pt);
 \draw (4,0) -- (4,-3pt);
 \draw (4,-0.1) node[below] {$T_{k-1}$};
 
  \draw (6,0) -- (6,3pt);
 \draw (6,0) -- (6,-3pt);
 \draw (6,-0.1) node[below] {$T_{k}$};
 
 \draw (6,0) node[above] {$\uparrow$};
  \draw (6,0.4) node[above] {$\tau_k(F(T_{k-1},T_{k})-K_F)$};
\end{tikzpicture}
\end{center}
\caption{Swaplet based on the forward-looking rate.}
\label{fig:FRA_forwardForward}
\end{figure}
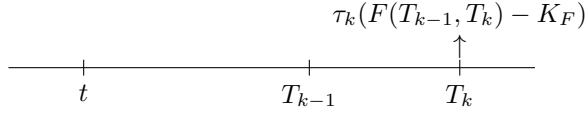

So we have defined two types of forwards: the forward of the backward-looking rate and the forward of the forward-looking rate. Nevertheless, for each $k=1,\ldots,N$, the backward-looking forward rate $R_k$ and the forward-looking forward rate $F_k$ can be modeled by a single rate, the forward of the backward-looking rate $R_k$. In fact, before the beginning of the application interval $[T_{k-1},T_k)$, the two forwards are the same and given by the process $R_k$. At time $T_{k-1}$, $R_k$ sets at the forward-looking spot rate. Note that in the old LMM, $R_k$ will end at time $T_{k-1}$. Instead, in the FMM the rate still exists and continues to evolve. In fact, it evolves until the time $T_k$, where it fixes to the backward-looking spot rate. After time $T_k$, it continues to exist, and it is a constant. 

\subsection{Computation of extended discount factors from forward rates values} \label{sec:compExtendedBondPrices}
In this subsection we summarize how to compute $P(T_i,T_j)$:
\begin{equation}  \label{gammas} 
P(T_i,T_j) = 
     \begin{cases}
        \displaystyle \prod_{k=i+1}^{j} \dfrac{1}{1+\tau_k R_k(T_i)} &\mbox{ if } T_i<T_j,\\
       1 & \mbox{ if }T_i=T_j,\\
       \displaystyle\prod_{k=j+1}^i \big( 1 + \tau_k R_k(T_k) \big) &\mbox{ if } T_i>T_j.
     \end{cases}
\end{equation}
We can abridge these three cases in the following formula $$P(T_i,T_j) = P(T_i,T_0) \prod_{k=1}^j \dfrac{1}{1+\tau_k R_k(T_i)},$$ with the equality holding for $T_i=T_j$, $T_i<T_j$ or $T_i>T_j$, and being valid for $j>0$. The case $j=0$ is just the bank account, i.e., 
$$B(T_i) = P(T_i,T_0)= \prod_{k=1}^i \big( 1 + \tau_k R_k(T_k) \big).$$

\section{The generalized FMM}\label{sec:FMM}

In order to create a proper market model, it is not enough to model just the evolution of a single rate. Indeed, we need to model the evolution of the forward rates jointly. This is the objective of this section: we model the evolution of the forward rates under a common probability measure.

Also in the FMM, we can specify the forward rate dynamics under the classic spot-LIBOR measure $\Q^d$ and the general $T_k$-forward measure $\Q^{T_k}$. In fact, the FMM dynamics under $\Q^d$ and $\Q^{T_k}$ are the same as those of the corresponding LMM (see \cite{brigoMercurio}, for example).

Thanks to the Definition \ref{def:extendedBondPrice} of extended bond prices, the new FMM allows also for forward-rates dynamics under the risk-neutral measure $\Q$. This is a great advantage, which was not possible in the old LMM world. In fact, as previously stated, the bank account is a zero maturity zero coupon bond, i.e. a bond that instantaneously matures and directly transforms into the bank account, $P(t,0)=B(t)$. So, the risk-neutral measure $\Q$, associated with the bank account as numeraire, can be seen as a forward measure, i.e. the forward measure associated with the zero maturity bond. Therefore, one can derive $\Q$ dynamics for the forward rates. 

From now on in this article, we will restrict ourselves to the risk-neutral measure $\Q$.
Each forward rate $R_k$, $k=1,\ldots,N$, is modelled as a continuous time stochastic process $R_k(t)$. The dynamics of the forward processes are driven by a $N$-dimensional correlated Wiener processes $W_1^\Q(t),\ldots,W_N^\Q(t)$ under measure $\Q$. In order to streamline the notation along the article, we let $ W_k(t) = W^\Q_k(t)$, for $k=1,\ldots,N$. Let $\rho_{ij}$ denote the correlation coefficient between $W_i(t)$ and $W_j(t)$, i.e $\E[\rd W_i(t) \rd W_j(t)] = \rho_{ij} \rd t$, where $\rd W_k(t)$ denotes the increment of the Wiener process $W_k(t)$ under the measure $\Q$. The system of SDEs of the FMM takes the form
\begin{equation} \label{eq:SDEs_1_FMM}
 \rd R_k(t) = \mu_k(t) \rd t + \nu_k(t) \rd W_k(t), \quad k=1,\ldots,N,
\end{equation}
where $\mu_k(t)$ and $\nu_k(t)$ are the drift and diffusion terms of the forward rate $R_k(t)$, respectively. The drift terms are determined by requiring a lack of arbitrage. The diffusion terms have to capture the fact that the process $R_k(t)$ will not be killed at $t=T_{k-1}$ as it happened in the classic LMM. In the FMM we need to define dynamics of the forward rates $R_k(t)$ inside their application periods $[T_{k-1},T_k)$. Besides, the volatility of $R_k(t)$ inside $[T_{k-1},T_k)$ goes down progressively to zero: it becomes smaller and smaller until reaching the value zero at $T_k$. In order to model this behavior, the system \eqref{eq:SDEs_1_FMM} is modified in the following way
\begin{equation} \label{eq:SDEs_2_FMM}
 \rd R_k(t) = \mu_k(t) \rd t + \nu_k(t) \gamma_k(t) \rd W_k(t), \quad k=1,\ldots,N,
\end{equation}
where the new parameter $\gamma_k(t)$ incorporates the volatility decay in the model. Therefore, the volatility is decomposed in two components, one is the classic LMM volatility $\nu_k(t)$, while $\gamma_k(t)$ is a deterministic function to control the volatility decay. This function $\gamma_k$ is equal to one up to the beginning of the interval $[T_{k-1},T_k)$ and is equal to zero at $T_k$ to model the fact that the rate has expired. Moreover, $\gamma_k$ should be differentiable and decrease down to zero (going from one to zero) inside the interval $[T_{k-1},T_k)$. For example, in the Ho-Lee model, the volatility decay is linear inside the application period, thus $\gamma_k$ is defined as:
\begin{equation}  \label{gammas} 
\gamma_k(t) = 
     \begin{cases}
       1 & \mbox{ if }t\leq T_{k-1},\\
       \dfrac{T_k-t}{T_k-T_{k-1}} &\mbox{ if } t \in (T_{k-1},T_k), \\
       0 &\mbox{ if } t\geq T_k.
     \end{cases}
\end{equation}

Let $\sigma_k(\cdot)$ be a deterministic function of time $t$. Some standard volatility specifications are the following. In the so-called normal model, $\nu_k(t)=\sigma_k(t)$. For the lognormal model, $\nu_k(t)=\sigma_k(t)R_k(t)$, while for the shifted-lognormal model, which allows for negative rates, $\nu_k(t)=\sigma_k(t)(R_k(t)+\vartheta_k)$, where $\vartheta_k$ is a positive constant. For the CEV model $\nu_k(t)=\sigma_k(t)R_k(t)^{\beta_k}$, where $0\leq \beta_k \leq 1$. In the following, except where otherwise indicated, we will assume a general model specification, i.e. we let $\nu_k$ a general adapted process.

Under the measure $\Q^k$ associated with the numeraire $P(t,T_k)$, $R_k$ is a martingale, i.e. $R_k$ is the driftless process $\rd R_k(t) =  \nu_k(t) \gamma_k(t) \rd W_k^{\Q^k}(t)$, where $W_k^{\Q^k}$ denotes a Wiener process under measure $\Q^k$. In order to model all forward rates $R_k$, $k=1,\ldots,N$ jointly, the computation of the dynamics of each forward rate under a common probability measure is needed. In this work, we will consider as numeraire the bank account $B(t)=P(t,0)$. As previously stated, the probability measure associated with this numeraire is the risk-neutral measure $\Q$. 

Under the probability measure $\Q$ the price of the bonds $P(t,T_k)$ divided by the numeraire $B(t)=P(t,T_0)$ must be martingales. By using this condition, the drifts $\mu_k(t)$ for the forward rates can be computed starting from $R_1$ until $R_N$, thus obtaining (see \cite{lyashenkoMercurio:Mar2019} for details) 
\begin{equation} \label{eq:muk}
\mu_k(t) = \nu_k(t) \gamma_k(t) \sum_{i=1}^{k} \rho_{ik} \dfrac{\tau_i \nu_i(t) \gamma_i(t)}{1+\tau_i R_i(t)}.
\end{equation}
Since $\gamma_k(t)=0$ for $t\geq T_k$, $\mu_k$ can be better expressed in terms of the index function 
\begin{equation} \label{eq:indexFunction}
\eta(t) = \min\{j, T_j\geq t\},    
\end{equation}
which provides the index of the element in the tenor structure being not smaller than $t$ that is the nearest to time $t$. Therefore, we have
\begin{equation} \label{eq:muk_nu}
\mu_k(t) = \nu_k(t) \gamma_k(t) \sum_{i=\eta(t)}^{k} \rho_{ik} \dfrac{\tau_i \nu_i(t) \gamma_i(t)}{1+\tau_i R_i(t)}.
\end{equation}

All in all, the dynamics of $R_k$ under the measure $\Q$ satisfy the following system of SDEs:
\begin{align} \label{eq:dynamicsRkQ}
 \rd R_k(t) = \nu_k(t) \gamma_k(t) \sum_{i=\eta(t)}^{k} \rho_{ik} \dfrac{\tau_i \nu_i(t) \gamma_i(t)}{1+\tau_i R_i(t)} \rd t+ \nu_k(t) \gamma_k(t) \rd W_k(t), \,  k=1,\dots,N.   
\end{align}

\section{PDE for the generalised FMM under the risk neutral measure} \label{sec:PDEsFMM}
After the introduction of the generalized FMM in the previous section, where the dynamics of the forward rates $R_k$ satisfy the system (\ref{eq:dynamicsRkQ}) under the risk-neutral measure $\Q$, in this section we derive the corresponding PDE formulation for the pricing of interest rate derivatives without early exercise opportunity (also referred to as European interest rate derivatives). As in the case of more classical models for interest rate derivatives, such as LMM, or any other derivatives, the statement of the PDEs formulation is based on the appropriate Feynman-K\`ac theorem (see \cite{brigoMercurio}, for example). This theorem relates the formulation in terms of expectations with the one in terms of PDEs, so that the solution of the PDE represents the expectation of an appropriate process under a certain probability measure, in this case, we consider the risk-neutral one.
\begin{proposition}[FMM PDE] \label{prop:FMMPDE}
Let $\eta$ be the index function defined in \eqref{eq:indexFunction}.
Let $R^{min}\in\R$ be a potentially negative lower bound for the rates (to include the relevant cases of shifted-lognormal models). Let $\nu_k(t) = \nu_k(t,R_k(t))$ be a general instantaneous volatility for the forward rate $R_k(t)$.
Under the risk-neutral measure $\Q$, the price of an interest rate derivative with maturity $T=T_k>T_0=0$ (for some $k=1,\ldots,N$), that depends on the fixing of the rates $R_1,\ldots,R_N$, with payoff function $\varphi:[R^{min},\infty)^N \to \R$, is given by $$V(t,R_1,\ldots,R_N)=P(t,T_0) \Pi(t,R_1,\ldots,R_N), \quad t\in [T_0,T],$$
where the relative price $\Pi:[T_0,T]\times [R^{min},\infty)^N \rightarrow \R$ satisfies the PDE
\begin{equation} \label{eq:PDEFMM}
 \dfrac{\partial\Pi}{\partial t} + \sum_{k=1}^N \mu_k(t)\dfrac{\partial\Pi}{\partial R_k} + \frac12 \sum_{k,l=\eta(t)}^N \rho_{kl} \nu_k(t) \gamma_k(t) \nu_l(t) \gamma_l(t) \dfrac{\partial^2 \Pi}{\partial R_k \partial R_l} = 0,   \quad t\in[T_0,T),
\end{equation}
along with the terminal condition
$$\Pi(T,R_1,\ldots,R_N)=\dfrac{\varphi(R_1,\ldots,R_N)}{P(T,T_0)}, \quad R_1,\ldots,R_N\geq R^{min}.$$
\end{proposition}

\begin{proof}
The derivative pays out $\varphi(R_1(T),\ldots,R_N(T))$ at the maturity date $T$. Note that receiving an amount $X$ of money at time $T$ is equivalent to receiving an amount $\dfrac{X}{B(T)} = \dfrac{X}{P(T,T_0)}$ of $T_0$-bonds. Therefore, the payoff can be interpreted as a relative payoff, in the sense that it is a payoff in terms of an amount of extended zero-coupon bonds maturing at $T_0$. If $\Pi(t,R_1,\ldots,R_N)$ denotes the time $t$ relative  price of such a derivative, standard pricing theory states that
\begin{align} \label{eq:Pi}
    \Pi(t,R_1,\ldots,R_N) = & \dfrac{V(t,R_1,\ldots,R_N)}{P(t,T_0)} \\
    = & \E^\Q \left[ \dfrac{\varphi(R_1(T),\ldots,R_N(T))}{P(T,T_0)} \Bigg\lvert \mathcal{F}_t\right], \nonumber
\end{align}
which allows to write the relative price of the interest rate derivative in terms of an expectation under the risk-neutral probability measure.

We refer the reader to Section \ref{sec:compExtendedBondPrices} to check how to compute $P(T,T_0)$ in terms of $R_1(T), \ldots, R_N(T)$.
Having in mind that \eqref{eq:dynamicsRkQ} describes the system of SDEs under the risk neutral measure $\Q$, the corresponding infinitesimal generator $\mathcal{L}_{R_1,\ldots,R_N}$ of $R_1,\ldots,R_N$ is thus given by (see \cite{Oksendal}, for example):
\begin{align*}
    \mathcal{L}_{R_1,\ldots,R_N} &=\dfrac{\partial}{\partial t} + \sum_{k=1}^N \mu_k(t)\dfrac{\partial}{\partial R_k} + \frac12 \sum_{k,l=1}^N \rho_{kl} \nu_k(t) \gamma_k(t) \nu_l(t) \gamma_l(t) \dfrac{\partial^2}{\partial R_k \partial R_l}.
\end{align*}
Since $\gamma_k(t)$ and $\gamma_l(t)$ are zero for $k,l<\eta(t)$ the infinitesimal generator can be written as
\begin{align*}
    \mathcal{L}_{R_1,\ldots,R_N} &=\dfrac{\partial}{\partial t} + \sum_{k=1}^N \mu_k(t)\dfrac{\partial}{\partial R_k} + \frac12 \sum_{k,l=\eta(t)}^N \rho_{kl} \nu_k(t) \gamma_k(t) \nu_l(t) \gamma_l(t) \dfrac{\partial^2}{\partial R_k \partial R_l}.
\end{align*}
Applying Feynman-Kac theorem, if $\Pi$ satisfies the PDE
\begin{align} \label{eq:pdeL}
    \mathcal{L}_{R_1,\ldots,R_N} \Pi &= 0,\\
    \Pi(T,\cdot) &= \dfrac{\varphi(\cdot)}{P(T,T_0)},
\end{align}
then $\Pi$ satisfies equation \eqref{eq:Pi}.
\end{proof}

PDE \eqref{eq:PDEFMM} diffuses a relative price, i.e., a price in terms of a bond. After having numerically solved the PDE and thereby having obtained the time $t$ relative value function, the latter has to be multiplied by the time $t$ bond price $P(t,T_0)$ to obtain the absolute value price (the price of the derivative itself). Note that if $t=T_0$, since $P(T_0,T_0)=1$, then $V(T_0,R_1,\ldots,R_N)= \Pi(T_0,R_1,\ldots,R_N)$.

\section{Numerical methods}\label{sec:numericalMethods}

In this section, we design both stochastic and deterministic numerical me\-thods to price interest rate derivatives in the recently introduced FMM setting. Although there is a huge variety of products which could be priced, in this work we will consider swaptions. The main reason is that swaptions markets are one of the two main markets in the world of options on interest rates (along with caps and floors). However, the proposed methodology can be extended to a large variety of interest rate derivatives.

In this section and the forthcoming one about numerical results, for sim\-pli\-ci\-ty we will consider the lognormal model for volatilities, i.e., we assume that $\nu_k(t)=\sigma_k(t)R_k(t)$ (therefore $R^{min}=0$). This model is one of the most popular, and it allows for a more straightforward presentation of the proposed numerical methodologies.

\subsection{RFR swaptions} 

Let us start defining interest rate swaps (IRS). An IRS is a contract with a counterparty to exchange payments indexed to interest rates at future fixed dates.  At each time $T_i$, $i=a+1,\ldots,b$, the contract holder pays a fixed interest rate $K$ and receives the floating interest rate $R(T_{i-1},T_i)=R_i(T_i)$. Therefore, at time $T_a$ the value of the swap is given by
\begin{equation} \label{eq:irs_value}
  \text{IRS}(T_a; T_a,\ldots,T_b)= \sum_{i=a+1}^b P(T_a,T_i) \tau_i (R_i(T_a)-K).  
\end{equation}
As the underlying rates are RFRs, we can refer to it more precisely as an RFR swap.

Let us now define an option over an IRS with maturity at $T_a$ and where the length of the underlying IRS is $(T_b-T_a)$. We denote such a derivate as a swaption $T_a\times(T_b-T_a)$.
A swaption is an option giving the right to enter a swap at the swaption's maturity date $T_a$. More precisely, the RFR swaption payoff at time $T_a$ is given by
$$\max(\text{IRS}(T_a; T_a,\ldots,T_b),0).$$

Since backward-looking and forward-looking forward rates with the same application period are equal at every time before the start of the period, backward-looking and forward-looking swaptions have the same value.

\subsection{Monte Carlo pricing of swaptions with the FMM} \label{monteCarloSwaption}
In the next paragraphs, we describe how to price RFR swaptions under the FMM using Monte Carlo simulation. The results obtained by this method will serve as benchmark prices for the forthcoming PDE numerical solutions.

In terms of expectations, according to expression \eqref{eq:Pi}, the relative price of the swaption $T_a\times(T_b-T_a)$ at time $T_0$ is given by
\begin{equation} \label{exp_swaption}
\E^\Q\left[ \dfrac{\max(\text{IRS}(T_a; T_a,\ldots,T_b),0)}{P(T_a,T_0)} \right].
\end{equation}
Next, taking into account \eqref{eq:irs_value} and the computation of the discount factors in Section \ref{sec:compExtendedBondPrices}, the expectation \eqref{exp_swaption} depends on the joint distribution of the forward rates $R_j$, with $1\leq j \leq b$, at time $T_a$, i.e. $R_1(T_a), \ldots, R_a(T_a), R_{a+1}(T_a),\ldots,$ $R_{b}(T_a)$. According to the dynamics \eqref{eq:SDEs_2_FMM}, we need to generate several simulations of such rates at time $T_a$. Finally, we evaluate the relative payoff in each simulation and average. 

Since the dynamics \eqref{eq:SDEs_2_FMM} does not lead to a process with a known distribution, we perform a time discretization with the small time step $\Delta t$. Moreover, we introduce logarithmic interest rates in the FMM setting defined by \eqref{eq:SDEs_2_FMM}, so that by using Ito's formula we get the following equivalent formulation of the FMM:
\begin{align*}
    \rd \ln R_k(t) = \left(\mu_k(t) \dfrac{1}{R_k(t)} - \frac12 \gamma_k^2(t) \sigma_k^2(t)  \right) \rd t + \sigma_k(t) \gamma_k(t) \rd W_k(t).
\end{align*}
This formulation has the advantage that the diffusion coefficient $\sigma_k \gamma_k(t)$ is deterministic. Therefore,  Euler and Milstein schemes coincide. Consequently, the time discretization
\begin{align*}
    R_k(t+\Delta t) =& R_k(t) \exp\bigg(\mu_k(t) \dfrac{1}{R_k(t)} \Delta t - \frac12 \gamma_k^2(t) \sigma_k^2(t)  \Delta t + \\
    &\hspace{4.0cm}\sigma_k(t) \gamma_k(t) (W_k(t+\Delta t)- W_k(t)) \bigg),
\end{align*}
leads to an approximation of the exact process. Note that the Brownian motion increments are normally distributed with mean $0$ and standard deviation $\sqrt{\Delta t}$, and correlated with correlation matrix $\rho=(\rho_{kl})_{k,l=1,\ldots,N}$.

\newcommand{\ds}{\displaystyle}
\newcommand{\F}{{\cal F}}
\newcommand{\A}{{\cal A}}
\subsection{Numerical integration of the multi-dimensional PDE} \label{sec:amfr-w}
In this section we design an appropriate and efficient numerical scheme for solving the PDE stated in Section \ref{sec:PDEsFMM}, when considering the lognormal model for volatilities. 

Taking into account the notation of the tenor structure for the swap, let us focus on the pricing of the swaption $T_1\times (T_N-T_1)$,
with payoff function
$$\varphi(R_1,\dots,R_N)= \max\left\{ \ds\sum_{k=2}^N P(T_1,T_k) \tau_k(R_k(T_1)-K), \,0\right\},$$
so the relative price $\Pi(t,R_1,\dots,R_N)$ satisfies the PDE \eqref{eq:PDEFMM} for $t\in [T_0,T_1]$ and the final condition
$$\Pi(T_1,R_1,\dots,R_N) = \dfrac{1}{P(T_1,T_0)}  \max\left\{ \ds\sum_{k=2}^N P(T_1,T_k) \tau_k(R_k-K), \,0\right\}.$$
In order to solve numerically this PDE problem, for simplicity we consider the following change of time variable and notations:
\begin{equation}\label{change1}
u(t,x_1,\dots,x_N)= \Pi(T_1-t,R_1,\dots,R_N),\quad t\in [0,T_1-T_0]=[0,\tau_1],\,x_k=R_k\ge 0.    
\end{equation}

Note that we are allowing a certain abuse of notation: we maintain the notation $t$, which initially represented physical time and hereafter represents the remaining time up to $T_1$. In this way the final condition at physical time $T_1$ becomes an initial condition in the PDE formulation in the new time variable. Numerical results in the next section will always refer to the physical time.

After the previously indicated change of time variable and notations, the PDE problem can be expressed as follows, find the function $u$, such that:
\begin{equation}\label{edpdimN}
\begin{array}{rl}
\dfrac{\partial u}{\partial t} =& \ds\sum_{k=1}^N \mu_k(t,x_1,\dots,x_k)\dfrac{\partial u}{\partial x_k}  + \ds\sum_{k=1}^N \delta_k(t,x_k)  \dfrac{\partial^2 u}{\partial x_k^2} \\[0.8pc]
 +& \ds\sum_{k=1}^{N-1} \ds\sum_{l=k+1}^N \varrho_{kl}(t,x_k,x_l) 
\dfrac{\partial^2 u}{\partial x_k\,\partial x_l}  ,
\end{array}
\end{equation}
for all $x_k>0$, where
\begin{equation}\label{coeff_edp}
\begin{array}{l}
 \mu_k(t,x_1,\dots,x_k)= \lambda_k(t) x_k \,\ds\sum_{j=1}^k \rho_{kj} \lambda_j(t) \dfrac{\tau_j}{1+\tau_j x_j} x_j,  \\[0.5pc]
 \delta_k(t,x_k)= \frac12 \lambda_k^2(t) x_k^2,\\[0.5pc]
 \varrho_{kl}(t,x_k,x_l) =\rho_{kl} \lambda_k(t) \lambda_l(t) x_k x_l,\\[0.5pc]
 \lambda_k(t)= \sigma_k(T_1-t) \gamma_k(T_1-t)=\left\{\begin{array}{ll}
 \sigma_1(T_1-t) \dfrac{t}{\tau_1}, & \hbox{ if } k=1 \\[0.5pc]
  \sigma_k(T_1-t), & \hbox{ if } k\ge 2, 
 \end{array}\right. ,\quad t\in [0,\tau_1],
\end{array}
\end{equation}
with initial condition
\begin{equation}\label{ini_N}
\begin{array}{l}
u(0,x_1,\dots,x_N)=u_0(x_1,\dots,x_N):=\max \{g(x_1,\dots,x_N),0\}, \\[0.5pc]
g(x_1,\dots,x_N)=\ds\sum_{k=2}^N \left( \ds\prod_{l=1}^k \dfrac{1}{1+\tau_l x_l} \right) \tau_k (x_k-K).
\end{array}
\end{equation}

Note that the expression of $\lambda_k$ in (\ref{coeff_edp}) is valid for $t\in [0,\tau_1]$.

As usual, in financial problems initially posed in unbounded spatial domains, for the numerical integration of this PDE, the spatial domain must be restricted to a rectangle 
$(x_1,\dots,x_N)\in \Omega=(0,R_1^{max})\times \dots \times (0,R_N^{max})$, by selecting the values $\{R_k^{max}\}_{k=1}^N$  large enough. On the upper boundaries, a linear behavior of the solution is assumed by considering the conditions
\begin{equation}\label{bc_N}
\dfrac{\partial^2 u}{\partial x_k^2}(t,x_1,\dots,x_N) = 0 \quad \hbox{if } x_k=R_k^{max},\quad k=1,\dots,N,
\end{equation}
while, due to the degeneracy of the PDE at the boundaries $x_k=0$, no conditions are required at these boundaries.

Next, we propose a finite difference method to approximate the solution of the PDE problem \eqref{edpdimN}-\eqref{coeff_edp}-\eqref{ini_N}-\eqref{bc_N}. 
Firstly, a discretization of the spatial derivatives $u_{x_k},u_{x_k x_k},u_{x_k x_j}$ must be done. In order to build a spatial mesh, it must be taken into account that the payoff \eqref{ini_N} lacks differentiability at the points of $\Omega$ such that
\begin{equation}\label{areaint}
\tilde{g}(x_2,\dots,x_N):=\ds\sum_{k=2}^N \left( \ds\prod_{l=2}^k \dfrac{1}{1+\tau_l x_l} \right) \tau_k (x_k-K)=0,
\end{equation} 
so, as it is recommended in \cite[Chap. 4]{KarelintHout17} and the references therein, a non-uniform spatial mesh on each $x_k-$direction, for $k\ge 2$, results more convenient. Therefore, given $N$ positive integers $\{M_1,\dots,M_N\}$ we consider the spatial grid
\begin{equation}\label{mesh1}
\hspace{-2em}\begin{array}{rl}
\hbox{for } x_1:&   x_{1,j_1}=j_1 h_1,\quad 0\le j_1\le M_1, \quad h_1=\dfrac{R_1^{max}}{M_1},\\[0.5pc]
\hbox{for } x_k,\,k\ge 2:& x_{k,j_k}=K+L_k \sinh \xi_{k,j_k},\, 0\le j_k\le M_k, \,h_{k,j_k}= x_{k,j_k}-x_{k,j_k-1},  \\[0.5pc]
& \xi_{k,j_k}=\xi_k^{min} + j_k \Delta \xi_k,\quad \Delta \xi_k=\dfrac{\xi_k^{max}-\xi_k^{min}}{M_k}, \\[0.5pc]
 & \xi_k^{min} = \sinh^{-1} (-K/L_k), \quad \xi_k^{max}=\sinh^{-1} \left((R_k^{max}-K)/L_k\right),
\end{array}
\end{equation}
where the parameters $L_k$ measure the fraction of grid points that are closer to $K$. In this case, we will consider $L_k=K/10,\,\forall k\ge 2.$

In addition, a cell averaging technique is applied to smooth the payoff at the grid points close to the hyperplane defined by \eqref{areaint}, by adapting the technique proposed in \cite[Chap. 4]{KarelintHout17} for 1D-PDEs. In this case, for each subset of indices $(j_3,\dots,j_N)$, $0\le j_k\le M_k,\,k\ge 3,$ we compute the value
%
$$\tilde{x}_2:= K-\dfrac{1}{\tau_2} \ds\sum_{k=3}^N \left( \ds\prod_{l=3}^k  \dfrac{1}{1+\tau_l x_{l,j_l}}\right) \,\tau_k(x_{k,j_k}-K),$$
($\tilde{x}_2=K$ for $N=2$)  and look for the index $j_{ind}$, with $0\le j_{ind} \le M_2$, such that
$|x_{2,j_{ind}}-\tilde{x}_2| = \min_{0\le j_2\le M_2} |x_{2,j_2}-\tilde{x}_2|.$ Then, the cell $[x_2^-,x_2^+]$ is considered, where
$$x_2^-=\dfrac{x_{2,j_{ind}-1}+x_{2,j_{ind}}}{2},\quad x_2^+=\dfrac{x_{2,j_{ind}}+x_{2,j_{ind}+1}}{2},\quad \tilde{h}_2=x_2^+ -x_2^- ,
$$
and the initial condition at the points $(x_{1,j_1},x_{2,j_{ind}},x_{3,j_3},\dots,x_{N,j_N})$, for all $j_1=0,\dots,M_1$, is taken as the average over this cell $[x_2^-,x_2^+]$, that is,
$$\begin{array}{l}
u(0,x_{1,j_1},x_{2,j_{ind}},x_{3,j_3},\dots,x_{N,j_N})=\dfrac{1}{\tilde{h}_2} \ds\int_{x_2^-}^{x_2^+} u_0(x_{1,j_1},x_{2},x_{3,j_3},\dots,x_{N,j_N})\,dx_2 \\[0.8pc] 
\hspace{5em} =\dfrac{1}{\tilde{h}_2 (1+\tau_1 x_{1,j_1})} \left(  x_2^+ -\tilde{x}_2  - \left( K+ \dfrac{1-H_3}{\tau_2} \right)  \log \left( \dfrac{1+\tau_2 x_2^+}{1+\tau_2 \tilde{x}_2} \right) \right),
\end{array}$$
with
$$H_3 = \left\{ \begin{array}{ll}
0, & \hbox{ if } N=2, \\
\ds\sum_{k=3}^N \left( \ds\prod_{l=3}^k \dfrac{1}{1+\tau_l x_{l,j_l}} \right) \tau_k (x_{k,j_k}-K)  & \hbox{ if } N\ge 3. 
\end{array}\right. $$
%
%
%
On the other mesh points, the payoff function given in \eqref{ini_N} is applied.

Following the notation given in \cite[Sec. 3]{LopezPerezVazquez:sisc}, the grid points are rearranged as
$$\overline{\Omega}_h=\{ \vec{x}_J=(x_{1,j_1},\dots,x_{N,j_N}) \,:\, J=\vartheta_0(j_1,\dots,j_N),\,0\le j_k\le M_k, \,1\le k\le N\},$$ 
where the bijection $\vartheta_0: \, {\cal I}_N^{(0)} \rightarrow \{0,1,\dots,M-1\}$, with $M=\ds\prod_{k=1}^N (M_k+1)$ is defined as
$$\vartheta_0(x_{1,j_1},\dots,x_{N,j_N})= \ds\sum_{k=1}^N j_k E_k,\quad \mbox{with} \quad E_1=1,\,\,E_k=\ds\prod_{l=1}^{k-1} (M_l+1),\,k\ge 2 .$$
Because of the boundary conditions, the solution of the PDE on all the nodes (including the boundary points) of the spatial grid is unknown. 
Then, by approxi\-ma\-ting the spatial derivatives in the PDE by second-order central finite differen\-ces schemes on all the nodes of the spatial grid, for each $J=0,1,\dots,M-1$, $(j_1,\dots,j_N)=\vartheta_0^{-1}(J)$, we get
\begin{equation}\label{discret1}
\begin{array}{rl}
Y_J'=& \ds\sum_{k=1}^N \mu_k(t,x_{1,j_1},\dots,x_{k,j_k}) \nabla_J^{(k)} 
+ \ds\sum_{k=1}^N \delta_k(t,x_{k,j_k}) \Delta_J^{(k)}  \\[0.5pc]
&+ 
\ds\sum_{k=1}^{N-1} \ds\sum_{l=k+1}^N \varrho_{kl}(t,x_{k,j_k},x_{l,j_l}) \Delta_J^{(kl)}, 
\end{array}
\end{equation}
where the involved discrete operators are 
\begin{equation}\label{discret2}
\begin{array}{l}
\nabla_J^{(1)}=\dfrac{Y_{J+E_1}-Y_{J-E_1}}{2h_1},\quad \nabla_J^{(k)}=\beta_{j_k,-1}^{(k)} Y_{J-E_k}+\beta_{j_k,0}^{(k)} Y_J+ \beta_{j_k,1}^{(k)} Y_{J+E_k},\,\, k\ge 2, \\[0.5pc]
\Delta_J^{(1)}=\dfrac{Y_{J+E_1}-2Y_J+Y_{J-E_1}}{h_1^2} ,\\[0.5pc]
\Delta_J^{(k)}=\eta_{j_k,-1}^{(k)} Y_{J-E_k}+\eta_{j_k,0}^{(k)} Y_J+ \eta_{j_k,1}^{(k)} Y_{J+E_k},\, k\ge 2,\\[0.5pc]
\Delta_J^{(kl)}=\beta_{j_l,-1}^{(l)} \nabla^{(k)}_{J-E_l}+\beta_{j_l,0}^{(l)} \nabla^{(k)}_J+ \beta_{j_l,1}^{(l)} \nabla^{(k)}_{J+E_l},\,  l\ge k+1,
\end{array}
\end{equation}
with
\begin{equation}\label{diff_coef}
\begin{array}{ll}
\beta_{j_k,-1}^{(k)}=\dfrac{-h_{k,j_k+1}}{h_{k,j_k}(h_{k,j_k}+h_{k,j_k+1})}, &
\eta_{j_k,-1}^{(k)}=\dfrac{2}{h_{k,j_k}(h_{k,j_k}+h_{k,j_k+1})}, \\[0.5pc]
\beta_{j_k,0}^{(k)}=\dfrac{h_{k,j_k+1}-h_{k,j_k}}{h_{k,j_k}h_{k,j_k+1}}, &
\eta_{j_k,0}^{(k)}=\dfrac{-2}{h_{k,j_k}h_{k,j_k+1}}, \\[0.5pc]
\beta_{j_k,1}^{(k)}=\dfrac{h_{k,j_k}}{h_{k,j_k+1}(h_{k,j_k}+h_{k,j_k+1})}, &
\eta_{j_k,1}^{(k)}=\dfrac{2}{h_{k,j_k+1}(h_{k,j_k}+h_{k,j_k+1})}.
\end{array}
\end{equation}
It must be observed that, with this procedure, the spatial discretization is expanded to the borders. On the one hand, at the points of the ``right" borders the conditions \eqref{bc_N} have to be involved, so when $j_k=M_k$,  $\Delta_J^{(k)}=0$ is imposed. Therefore, virtual points $Y_{J+E_k}$ must be defined when $j_k=M_k$,
$$Y_{J+E_1}=2Y_J -Y_{J-E_1},\quad  Y_{J+E_k}= -\dfrac{\eta_{j_k,-1}^{(k)}}{\eta_{j_k,1}^{(k)}}  Y_{J-E_k} - 
\dfrac{\eta_{j_k,0}^{(k)}}{\eta_{j_k,1}^{(k)}}  Y_J,\,k\ge 2.
$$
On the other hand, for each $k=1,\dots,N,$ on the boundary $x_k=0$, we get $\mu_k=\delta_k=0,\, \forall k,$ $\varrho_{kl}=0,\,\forall l\ge k+1$. As a consequence, $\nabla_J^{(k)}=\Delta_J^{(k)}=0,\, \forall k,$ $\Delta_J^{(kl)}=0,\,l\ge k+1,$ when $j_k=0$.

By grouping all the approximations in a vector $Y=(Y_J)_{J=0}^{M-1}$, these equations can be written as the initial value problem with a directional splitting
\begin{equation}\label{ivp2d}
\begin{array}{c}
Y'={\cal F}(t,Y)=\ds\sum_{k=0}^N {\cal F}_k(t,Y), \quad Y(0)=Y_0,\\
{\cal F}_k(t,Y)={\cal A}_k(t) Y,\quad k=0,1,\dots,N,\\
{\cal A}_1(t)= \lambda_1^2(t) \tilde{{\cal A}}_1 ,\quad {\cal A}_k(t) =\lambda_k^2(t) \tilde{{\cal A}}_k^{(1)}+\lambda_k(t) {\cal D}_k(t) \tilde{{\cal A}}_k^{(2)},\,\,k=2,\dots,N,\\
\end{array}
\end{equation}
where each $\F_k(t,Y)$ stores the components of the discretization of the advection and diffusion terms in the $x_k-$direction, for $k=1,\dots,N,$ and $\F_0(t,Y)$ stores those of the discretization of the mixed derivatives. In this case, $\tilde{{\cal A}}_1$, $\{\tilde{{\cal A}}_k^{(1)},\tilde{{\cal A}}_k^{(2)}\}_{k=2}^N$ are block tridiagonal constant matrices and ${\cal D}_k(t)$ is diagonal. For the sake of simplicity, the coefficients of these matrices are given in \ref{ap:matrixCoefs}.

Due to the increasing stiffness of \eqref{ivp2d} as the resolution of the spatial grid increases, explicit methods are not suitable for its time integration. On the other hand, fully implicit methods requiring the computation of the exact Jacobian of the derivative function are also unsuitable because of the complicated structure of the matrix ${\cal A}_0(t)$. Therefore, for the time integration of \eqref{ivp2d} a method from the class of AMFR-W-methods (\cite[Sec. 4]{GlezHairerHdezPerez18}, \cite{GlezHdezPerez21}) is applied. In particular, we have selected the one-stage AMFR-W1 method. More precisely, given an approximation $Y_n$ to the solution of \eqref{ivp2d} at the time $t=t_n$, this method approximates the solution at $t=t_{n+1}=t_n+\Delta t$ (with $\Delta t$ being the constant step of the time discretization) by 
\begin{equation}\label{amfrw1}
\begin{array}{rl}
K^{(0)}=& \Delta t \, {\cal F}(t_n,Y_n), \\
(I-\nu \Delta t \,  {\cal A}_k(t_n))K^{(k)}=& K^{(k-1)}+\nu (\Delta t )^2 \alpha_{k,n},\quad k=1,\dots,N, \\
\tilde{K}^{(0)}=& 2K^{(0)}+\theta (\Delta t)^2 G_n- (I-\theta \Delta t \, {\cal A}(t_n))K^{(N)}, \\
(I-\nu \Delta t \, {\cal A}_k(t_n))\tilde{K}^{(k)}=& \tilde{K}^{(k-1)}+\nu (\Delta t )^2 \alpha_{k,n},\quad k=1,\dots,N, \\
Y_{n+1}=& Y_n+ \tilde{K}^{(N)},
\end{array}
\end{equation}
where
$$\begin{array}{l}
{\cal A}(t_n)=\dfrac{\partial {\cal F}}{\partial Y}(t_n,Y_n)=\ds\sum_{k=0}^N {\cal A}_k(t_n),\\[0.8pc]
 \alpha_{k,n}=\dfrac{\partial {\cal F}_k}{\partial t}(t_n,Y_n), \,k=1, \dots, N, \quad G_{n}=\dfrac{\partial {\cal F}}{\partial t}(t_n,Y_n) ,
\end{array}
$$
with parameters $\theta=1/2$ and $\nu=\theta$ for $N=2,3$ and $\nu= \kappa_N N \,\theta$ for $N\ge 4$, where the values of $\kappa_N$ are given in \cite[Table 2]{GlezHairerHdezPerez18} and guarantee that the AMFR-W1 method is unconditionally stable on multi-dimensional linear constant coefficient PDEs with mixed derivatives. 

It must be observed that, due to the block tridiagonal structure of the matrices $\A_k(t),\,k=1,\dots,N$, the linear systems with coefficient matrix $(I- \nu \Delta t \, {\cal A}_k(t_n))$ of dimension $M$ in the method are decoupled in tridiagonal systems of dimension $M_k+1$, which drastically reduces their computational cost (see details in \ref{ap:matrixCoefs}). In addition, it is not necessary to compute the full Jacobian ${\cal A}(t_n)$ in \eqref{amfrw1}, which does not have a block structure that can reduce its computational cost, due to the presence of the discretization of the mixed derivatives. The product ${\cal A}(t_n) K^{(N)}$ in the right-hand side in the definition of $\tilde{K}^{(0)}$ in \eqref{amfrw1} is obtained by an additional evaluation of the derivative function, since
$${\cal A}(t_n) K^{(N)}= \left(\ds\sum_{k=0}^N \A_k(t_n) \right) K^{(N)} =\F(t_n,K^{(N)}).$$

\section{Numerical experiments} \label{sec:numericalExperiments}

In this section, we present some numerical results that correspond to several RFR swaptions to assess the correctness of the proposed numerical methods and the performance of the models. The selected swaptions aim to illustrate the behavior of the methodology when increasing the number of involved RFRs.

Some of the employed market data are shown in Table \ref{tab:marketData}, where constant volatilities are taken into account. Also, we will consider that $T_0=0$ and that the pricing date is the physical time $t=T_0=0$ (which will correspond to the solution of the PDE at time $t=T_1$ in the new time variable with the previously mentioned abuse of notation, although we will always refer to the physical time $t=0$).

\begin{table}[!h]
\begin{center}
\begin{tabular}{ |c|c|c|c| } 
 \hline
  $k$ & $T_k$ & $R_k(0)$ & $\sigma_k(t)$ \\ 
 \hline
 \hline
 $1$ & $0.25$ & $0.01$ & $0.2$ \\ 
 \hline
 $2$ & $0.5$ & $0.013$ & $0.15$ \\ 
 \hline 
 $3$ & $0.75$ & $0.014$ & $0.25$ \\ 
 \hline  
 $4$ & $1.0$ & $0.015$ & $0.26$ \\ 
 \hline   
 $5$ & $1.25$ & $0.016$ & $0.27$ \\ 
 \hline    
\end{tabular}
\caption{Hypothetical market data of RFRs for the numerical examples.}
\label{tab:marketData} 
\end{center}
\end{table}

Moreover, the constant correlation coefficients $\rho_{kl}=0.5$ have been chosen, for all $k,l=1,\ldots,N$, with $k\neq l$. 

In all forthcoming examples, we will start computing prices using Monte Carlo simulation, which will serve as benchmark RFR swaption prices for the corresponding PDE numerical solutions. 

The numerical experiments have been performed with the following hardware and software configurations: an AMD Ryzen 9 5950X 16-Core Processor with 128 GBytes of RAM, CentOS Linux, and GNU C++ compiler. We have not used a numerical linear algebra software package, because the resulting linear systems of equations are tridiagonal, as previously stated. Thus, the codes developing the numerical methods have been implemented from scratch. Besides, parallel computing was not considered, since the AMFR-W1 time integrator is highly sequential.

\subsection{2-dimensional case ($N=2$)}

In this first example, we price, at time $t=T_0=0$, several RFR European swaptions $T_1\times(T_2-T_1)$ for the values of the spot forward rates $R_1(0)$ and $R_2(0)$ given on Table \ref{tab:marketData}. More precisely, we value the swaption at-the-money (ATM) ($K=K_{ATM}$), a couple of swaptions out-of-the-money (OTM) ($K=1.1 K_{ATM}$, $K=1.2 K_{ATM}$), and two swaptions in-the-money (ITM)  ($K=0.8 K_{ATM}$, $K=0.9 K_{ATM}$).

Monte Carlo and PDE results are shown in Table \ref{tab:mcVSpde2d}. Monte Carlo confidence intervals have been obtained with $10^7$ simulations and 100 time steps for the Milstein discretization scheme. In order to check the performance of the PDE AMFR-W1 numerical method, firstly we have tested it on a spatial grid with $M_1=M_2=1024$ and a small constant time step size $\Delta t=\tau_1/2^{11}$, so that the error associated to the time integration is negligible when compared to the one due to the spatial discretization. PDE prices for the values of Table \ref{tab:marketData} are obtained by using multi-linear interpolation since those points could not belong to the spatial mesh \eqref{mesh1}. Note that both Monte Carlo and PDE prices are consistent, thus validating both numerical techniques and also the well-posedness of the PDEs formulation and the associated boundary conditions. Besides, the implied volatilities corresponding to the PDE prices are also shown. Since the chosen dynamics are lognormal, implied volatilities are flat (not perfectly flat since the swap rate is not exactly lognormal).

\begin{table}[!h]
\begin{center}
\begin{tabular}{|c|c|c|c|}
    \hline
    \multicolumn{4}{|c|}{Swaption $T_1\times (T_2-T_1)$} \\
    \hline
    $K$ & Monte Carlo Confidence Interval & PDE & Impl vol \\
    \hline
    $1.2\, K_{ATM}$ & $[6.569174\times 10^{-7},6.705475\times 10^{-7}]$ & $6.610817\times 10^{-7}$ & $0.150103$ \\
    \hline
    $1.1\, K_{ATM}$ & $[1.229203\times 10^{-5},1.235655\times 10^{-5}]$ & $1.230812 \times 10^{-5}$ & $0.150014$ \\
    \hline
    $K_{ATM}$ & $[9.663654\times 10^{-5},9.681989\times 10^{-5}]$ & $9.666517 \times 10^{-5}$ & $0.150003$ \\
    \hline
    $0.9 \, K_{ATM}$ & $[3.313149\times 10^{-4},3.315975\times 10^{-4}]$ & $3.314849 \times 10^{-4}$ & $0.150035$ \\
    \hline
    $0.8 \, K_{ATM}$ & $[6.460959\times 10^{-4},6.463961\times 10^{-4}]$ & $6.463699 \times 10^{-4}$ & $0.150143$ \\
    \hline
    Time  & $73.32$ $s$ & \multicolumn{2}{c|}{$603.82$ $s$, $M_1=M_2=1024$} \\
    \hline    
\end{tabular}
\caption{95\% Monte Carlo confidence intervals for swaption prices, $10^7$ simulations with a $100$ time steps. PDE prices, computed on the non-uniform grids, and corresponding implied volatilities.}
\label{tab:mcVSpde2d}
\end{center}
\end{table}

For the sole purpose of illustrating the behavior of the PDE solution,
in Fig. \ref{fig:solpde} the approximated values of $u$ at $t=T_1$ are reported in the same case with $M_1=M_2=1024$ mesh points. It must be emphasized again that, due to the time reversal \eqref{change1}, the values of $u$ when $t=T_1$ are equivalent to that of $\Pi$ when
$t=T_0=0$. The solution on the computational domain and a zoom in the area of financial interest near the strike rate are shown.

\begin{figure}[h]
\begin{center}
\hspace{-1em}\includegraphics[scale=0.32]{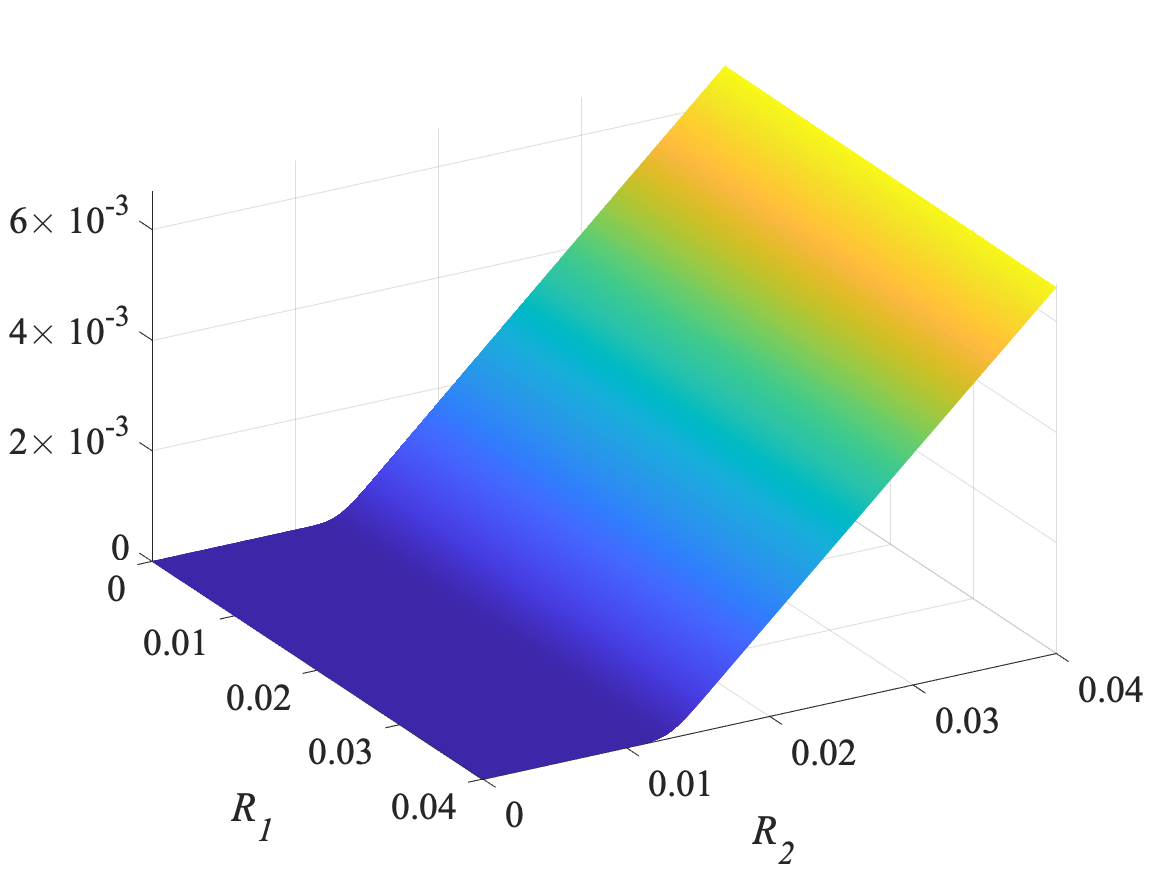}%
\includegraphics[scale=0.32]{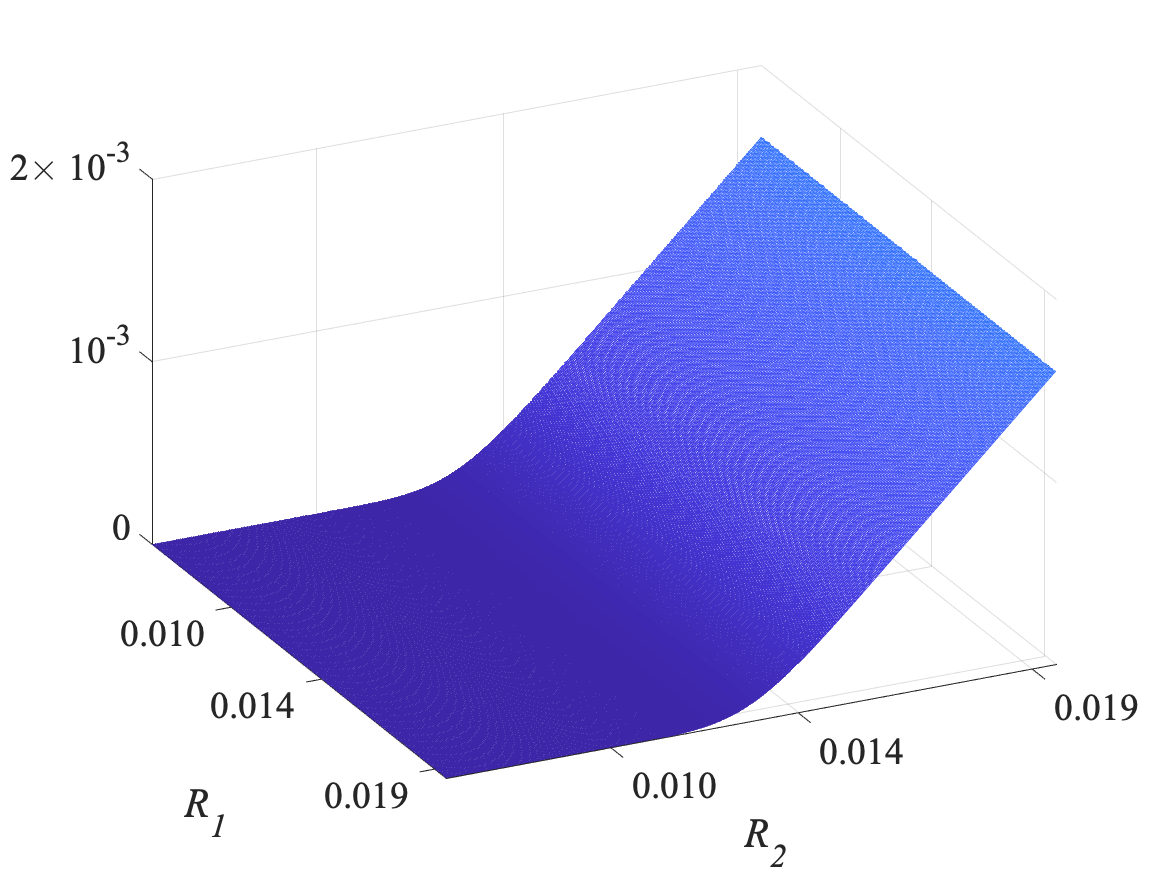}
\end{center}
\caption{PDE prices for the  $T_1\times(T_2-T_1)$ swaption, with $1024\times 1024$ grid. Left: full computational domain. Right: zoom in the area of interest $[0.5K_{ATM},1.5K_{ATM}]^2$.}\label{fig:solpde}
\end{figure}

In order to illustrate the convergence of the finite diffe\-ren\-ce method with AMFR-W1 \eqref{amfrw1}, we have applied it for several spatial meshes \eqref{mesh1}, with $M_1=M_2=L=2^r,\,r=2,3,\dots,10,$ and the same constant time step size $\Delta t=\tau_1/2^{11}$ as above. In Table \ref{tab:spaterr} the errors obtained for all cases are displayed, in both $l_2-$ and $l_{\infty}-$norms. There, the errors ($l_2-$error and $l_{\infty}$-error) have been computed with respect to the numerical solution given by the same method with $M_1=M_2=2^{11}$, and the estimated spatial orders ($l_2-$order and $l_{\infty}$-order) gathered in the last two columns are calculated by the well-known formula $$l_{p}-\mbox{order} = \frac{\log( \parallel l_p-\hbox{error}(2m) \parallel / \parallel l_p-\hbox{error}(m) \parallel)}{\log(2)},$$ for $m=2^r,\,r=2,3,\dots,9$, with $p=2$ and $p= \infty$. The second order of convergence is reflected in Table \ref{tab:spaterr}. 

In addition, in Table \ref{tab:spaterrunif} the results obtained using a uniform spatial mesh $x_{k,j_k}=j_k h_k,\,0\le j_k\le L$, with $h_k=R_k^{max}/L$, are shown. 
Although also achieving second-order convergence with uniform grids, the advantage of using non-uniform meshes is clearly illustrated when comparing Table \ref{tab:spaterr} and Table \ref{tab:spaterrunif}, as it improves the accuracy achieved by the method by two orders of magnitude.

\begin{table}[!h]
\begin{center}
\begin{tabular}{ |c|c|c|c|c|c| } 
 \hline
 $L$ & $l_2-$error & $l_{\infty}-$error & $l_2-$order & $l_{\infty}-$order & Time ($s$) \\ 
 \hline
    4 & $4.53\times 10^{-05}$& $1.02\times 10^{-04}$  &  -       & -  & $1.19 \times 10^{-02}$  \\ 
    8 & $3.76\times 10^{-06}$& $1.04\times 10^{-06}$  & $3.59$ & $3.29$ & $3.53 \times 10^{-02}$  \\ 
   16 & $4.00\times 10^{-07}$& $1.07\times 10^{-06}$ & $3.23$ & $3.28$ & $0.12$  \\ 
   32 & $1.07\times 10^{-07}$& $2.71\times 10^{-07}$ & $1.91$ & $1.98$ & $0.45$  \\ 
   64 & $2.59\times 10^{-08}$& $6.45\times 10^{-08}$ & $2.04$ & $2.07$ & $1.73$  \\ 
  128 & $6.72\times 10^{-09}$& $1.68\times 10^{-08}$ & $1.95$ & $1.94$ & $7.62$  \\ 
  256 & $1.67\times 10^{-09}$& $4.15\times 10^{-09}$ & $2.01$ & $2.02$ & $31.86$ \\ 
  512 & $3.97\times 10^{-10}$& $9.87\times 10^{-10}$ & $2.07$ & $2.07$ & $138.07$ \\ 
 1024 & $7.93\times 10^{-11}$& $1.97\times 10^{-10}$ & $2.32$ & $2.32$ & $603.82$ \\ 
 
 \hline
\end{tabular}
\caption{ATM Swaption $T_1\times(T_2-T_1)$: spatial errors and estimated orders on the non-uniform grids \eqref{mesh1} with $L=M_1=M_2$.}
\label{tab:spaterr} 
\end{center}
\end{table}

\begin{table}[!h]
\begin{center}
\begin{tabular}{ |c|c|c|c|c|c| } 
 \hline
 $L$ & $l_2-$error & $l_{\infty}-$error & $l_2-$order & $l_{\infty}-$order & Time ($s$)\\ 
 \hline
    4 & $4.30 \times 10^{-6}$& $8.33 \times 10^{-6}$ & - &- & $1.21\times 10^{-02}$ \\ 
    8 & $5.53 \times 10^{-6}$& $1.57 \times 10^{-5}$ & $-0.36$ & $-0.91$ & $3.59\times 10^{-02}$ \\ 
   16 & $3.87 \times 10^{-6}$& $1.23 \times 10^{-5}$ & $0.52$ & $0.35$ & $0.12$ \\ 
   32 & $1.13 \times 10^{-6}$& $5.00 \times 10^{-6}$ & $1.77$ & $1.30$ & $0.45$ \\ 
   64 & $3.21 \times 10^{-7}$& $1.95 \times 10^{-6}$ & $1.82$ & $1.36$ & $1.72$ \\ 
  128 & $6.91 \times 10^{-8}$& $3.03 \times 10^{-7}$ & $2.22$ & $2.68$ & $7.61$ \\ 
  256 & $1.98 \times 10^{-8}$& $1.13 \times 10^{-7}$ & $1.80$ & $1.42$ & $31.73$ \\ 
  512 & $4.13 \times 10^{-9}$& $1.88 \times 10^{-8}$ & $2.26$ & $2.59$ & $138.48$ \\ 
 1024 & $1.15 \times 10^{-9}$& $6.79 \times 10^{-9}$ & $1.85$ & $1.47$ & $605.42$ \\ 
 
 \hline
\end{tabular}
\caption{ATM Swaption $T_1\times(T_2-T_1)$: spatial errors and estimated orders on uniform grids with $L=M_1=M_2$.}
\label{tab:spaterrunif} 
\end{center}
\end{table}

\subsection{Higher-dimensional cases ($N=3,4,5$)}
Next, we focus on the pricing of higher dimensional RFR swaptions. More precisely, we consider the swaptions $T_1\times(T_3-T_1)$, $T_1\times(T_4-T_1)$ and $T_1\times(T_5-T_1)$ under the market data shown in Table \ref{tab:marketData}. Given the conclusions from the previous example about the numerical solution of the PDE, we just present the results for non-uniform grids.

PDE prices and the corresponding reference Monte Carlo confidence intervals are shown in Table \ref{tab:mcVSpde}. For each swaption, the number of mesh points in all directions is the same and denoted by $L$. As in the 2-dimensional case, to keep temporal errors negligible compared to spatial errors, for the time integration, the method AMFR-W1 \eqref{amfrw1} has been applied with constant step size $\Delta t=\tau_1/(2L)$. 
 As expected, all PDE prices lie inside the Monte Carlo confidence intervals. Note that as before, the 95\% Monte Carlo confidence intervals were computed by using $10^7$ simulations with 100 time steps in the Milstein time stepping.

\begin{table}[!h]
\begin{center}
\begin{tabular}{|c|c|c|c|}
    \hline
    \multicolumn{4}{|c|}{Swaption $T_1\times (T_3-T_1)$} \\
    \hline
    $K$ & Monte Carlo Confidence Interval & PDE & Impl vol \\
    \hline
    $1.2\, K_{ATM}$ & $[5.007571\times 10^{-6},5.070211\times 10^{-6}]$ & $5.020028\times 10^{-6}$ & $0.178879$ \\
    \hline
    $1.1\, K_{ATM}$ & $[4.532638\times 10^{-5},4.552660\times 10^{-5}]$ & $4.538339\times 10^{-5}$ & $0.177969$ \\
    \hline
    $K_{ATM}$ & $[2.361209\times 10^{-4},2.365753\times 10^{-4}]$ & $2.364758\times 10^{-4}$ & $0.177020$ \\
    \hline
    $0.9 \, K_{ATM}$ & $[7.014066\times 10^{-4},7.020817\times 10^{-4}]$ & $7.014788\times 10^{-4}$ & $0.176040$ \\
    \hline
    $0.8 \, K_{ATM}$ & $[1.340121\times 10^{-3},1.340854\times 10^{-3}]$ & $1.340742\times 10^{-3}$ & $0.175032$ \\
    \hline
    Time  & $112.94$ $s$ & \multicolumn{2}{c|}{$4316.30$ $s$,  $L = 256$} \\ 
    \hline
    \hline
    \multicolumn{4}{|c|}{Swaption $T_1\times (T_4-T_1)$} \\
    \hline
    $K$ & Monte Carlo Confidence Interval & PDE & Impl vol \\
    \hline
    $1.2\, K_{ATM}$ & $[9.480228 \times 10^{-6}, 9.589930\times 10^{-6}]$ & $9.523646\times 10^{-6}$ & $0.184582$ \\
    \hline
    $1.1\, K_{ATM}$ & $[7.775208 \times 10^{-5}, 7.808471\times 10^{-5}]$ & $7.788910\times 10^{-5}$ & $0.183922$ \\
    \hline
    $K_{ATM}$ & $[ 3.794420\times 10^{-4}, 3.801720\times 10^{-4}]$ & $3.800981\times 10^{-4}$ & $0.183272$ \\
    \hline
    $0.9 \, K_{ATM}$ & $[ 1.094727\times 10^{-3}, 1.095804\times 10^{-3}]$ & $1.095566\times 10^{-3}$ & $0.182621$ \\
    \hline
    $0.8 \, K_{ATM}$ & $[ 2.081112\times 10^{-3}, 2.082289\times 10^{-3}]$ & $2.082134\times 10^{-3}$ & $0.181977$ \\
    \hline
    Time  & $150.69$ $s$ & \multicolumn{2}{c|}{$23410.36$ $s$, $L=128$} \\    
    \hline
    \hline
    \multicolumn{4}{|c|}{Swaption $T_1\times (T_5-T_1)$} \\
    \hline
    $K$ & Monte Carlo Confidence Interval & PDE & Impl vol \\
    \hline
    $1.2\, K_{ATM}$ & $[ 1.485427 \times 10^{-5}, 1.501782 \times 10^{-5}]$ & $ 1.500055\times 10^{-5}$ & $0.188628$ \\
    \hline
    $1.1\, K_{ATM}$ & $[ 1.139641 \times 10^{-4}, 1.144421 \times 10^{-4}]$ & $1.143997 \times 10^{-4}$ & $0.187909$ \\
    \hline
    $K_{ATM}$ & $[ 5.350862 \times 10^{-4}, 5.361152 \times 10^{-4}]$ & $5.357548\times 10^{-4}$ & $0.187452$ \\
    \hline
    $0.9 \, K_{ATM}$ & $[ 1.515406 \times 10^{-3}, 1.516917 \times 10^{-3}]$ & $1.516010 \times 10^{-3}$ & $0.187002$ \\
    \hline
    $0.8 \, K_{ATM}$ & $[ 2.869551 \times 10^{-3}, 2.871208 \times 10^{-3}]$ & $ 2.870076\times 10^{-3}$ & $0.186816$ \\
    \hline
    Time  & $196.45$ $s$ & \multicolumn{2}{c|}{$77738.56$ $s$, $L=64$} \\ 
    \hline    
\end{tabular}

\caption{95\% Monte Carlo confidence intervals for swaption prices, $10^7$ simulations with a $100$ time steps. PDE prices and corresponding implied volatilities.}
\label{tab:mcVSpde}
\end{center}
\end{table}

Next, the numerical orders of convergence of the PDE method for the three and four-dimensional spatial examples are shown in Tables \ref{tab:spaterr_3d} and \ref{tab:spaterr_4d}, respectively. In such cases, the time integration has been performed with constant time step size $\Delta t=\tau_1/2^r$ with $r=9$ for $N=3$ and  $r=7$ for $N=4$.  For the five-dimensional case, we did not measure the order of convergence since this test would require a huge amount of RAM, which is not available in our machine.

\begin{table}[!h]
\begin{center}
\begin{tabular}{ |c|c|c|c|c|c| } 
 \hline
 $L$ & $l_2-$error & $l_{\infty}-$error & $l_2-$order & $l_{\infty}-$order & Time $(s)$\\ 
 \hline
8 	&  $1.30 \times 10^{-5}$ & $9.00 \times 10^{-5}$ & -& - & $0.19$\\
16 	&  $3.97 \times 10^{-6}$ & $3.76 \times 10^{-5}$ & $1.72$ & $1.26$ & $1.05$ \\
32 	&  $1.62\times 10^{-6}$ & $1.88\times 10^{-5}$ & $1.29$ & $1.00$ & $7.26$\\
64 	&  $3.15\times 10^{-7}$ & $4.25\times 10^{-6}$ & $2.37$ & $2.15$ & $58.23$ \\
128 &  $6.31\times 10^{-8}$ & $7.62\times 10^{-7}$ & $2.32$ & $2.48$ & $482.55$\\
 256 &  $1.24\times 10^{-8}$ & $1.46 \times 10^{-7}$ & $2.34$ & $2.39$ & $4316.30$\\
 \hline
\end{tabular}
\caption{ATM Swaption $T_1\times(T_3-T_1)$: spatial error and estimated orders on the non-uniform  grids \eqref{mesh1} with $L=M_1=M_2=M_3$.}
\label{tab:spaterr_3d} 
\end{center}
\end{table}

\begin{table}[!h]
\begin{center}
\begin{tabular}{ |c|c|c|c|c|c| } 
 \hline
 $L$ & $l_2-$error & $l_{\infty}-$error & $l_2-$order & $l_{\infty}-$order & Time ($s$) \\ 
 \hline
8 	& $ 6.96 \times 10^{-4}$ & $ 1.83 \times 10^{-2}$ & - & - & $0.59$  \\
16 	&  $ 1.45 \times 10^{-5}$ & $ 3.33\times 10^{-4}$ & $5.58$ & $5.78$ &  $6.62$ \\
32 	&  $ 2.35 \times 10^{-6}$ & $ 6.27 \times 10^{-5}$ & $2.63$ & $2.41$ & $97.18$  \\
64 	&  $ 3.91\times 10^{-7}$ & $ 1.53 \times 10^{-5}$ & $2.59$ & $2.03$ & $1473.89$  \\
 \hline
\end{tabular}
\caption{ATM Swaption $T_1\times(T_4-T_1)$:  spatial error and estimated orders on the non-uniform grids \eqref{mesh1} with $L=M_1=M_2=M_3=M_4$.}
\label{tab:spaterr_4d} 
\end{center}
\end{table}

Then, in Figure \ref{fig:times}, we present a graph of some of the previously shown execution times per problem size for both Monte Carlo and PDEs. As expected, Monte Carlo method does not suffer from the curse of dimensionality.

\begin{figure}[!h]
\begin{center}
\includegraphics[scale=0.5]{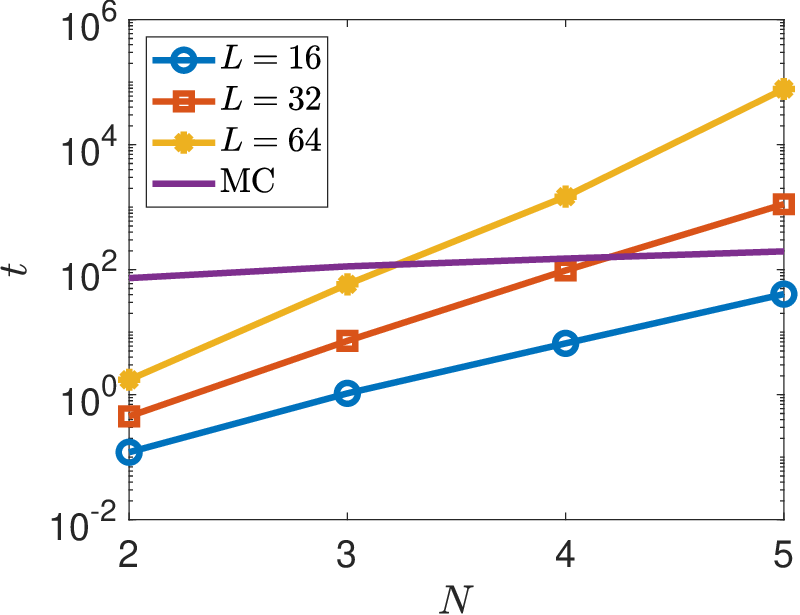}
\end{center}
\caption{Plot of execution times per problem size ($\log_{10}$ scale in the $y$ axis) . For the PDEs, the cases with $L=16,32,64$ are shown. For Monte Carlo method, $10^7$ paths were considered.}\label{fig:times}
\end{figure}

Finally, we show some numerical results for pricing Bermudan swaptions. We consider the case of two callability dates, known as Canary option. More precisely, we focus on the Canary swaption $T_3\times (T_4-T_3)$, where the two early exercises opportunities are $T_2$ and $T_1$, allowing to enter in the swaps $T_2\times(T_4-T_2)$ and $T_1\times(T_4-T_1)$, respectively. The results are shown in Table \ref{tab:canary}, along with the corresponding European swaption prices, which are lower as expected.

\begin{table}[!h]
\begin{center}
\begin{tabular}{ |c|c|c| } 
 \hline
 $K$ & European & Canary \\ 
  \hline
$1.2 K_{ATM}$ & $ 1.075029 \times 10^{-4}$& $ 1.205543 \times 10^{-4}$ \\ 
 \hline
$1.1 K_{ATM}$ & $1.938612 \times 10^{-4}$& $ 2.355759 \times 10^{-4}$ \\ 
 \hline
$K_{ATM}$ & $3.317602 \times 10^{-4}$& $4.483170 \times 10^{-4}$  \\ 
 \hline
$0.9 K_{ATM}$ & $5.339683 \times 10^{-4}$& $8.331122 \times 10^{-4}$  \\ 
 \hline
 $0.8 K_{ATM}$ & $8.032271\times 10^{-4}$& $1.553333 \times 10^{-3}$  \\ 
 \hline
\end{tabular}
\caption{PDE prices for the Swaption $T_3\times(T_4-T_1)$ considering $L=128$. Execution time is $73522.17$ seconds.}
\label{tab:canary} 
\end{center}
\end{table}


\section{Conclusions and future work} \label{sec:conclusions}
After the recent scandal of manipulation of IBORs, the worldwide financial authorities, and regulators started the replacement of IBORs by the so-called RFRs, that rely on real transactions.  In this so-called IBOR transition banks started to offer interest rate derivatives based on RFRs, thus motivating the need for an appropriate modeling RFRs dynamics to price these derivative products. The recent seminal rigorous article \cite{lyashenkoMercurio:Mar2019} introduced the generalized FMM. Pricing with Monte Carlo under FMM is very natural.

Having in view the limitations of Monte Carlo pricing techniques, in the present article we have rigorously stated for the first time in the literature a PDE formulation for pricing RFRs derivatives. Note that the spatial dimension is equal to the number of RFRs involved in the derivative payoff and can become large. In order to solve efficiently the PDE formulation, we propose the use of an AMFR-W1 finite differences method, which is specially appropriate and efficient to cope with the presence of mixed derivatives in the spatial variables. The numerical results illustrate the correctness of the method, by comparing them with a reference solution obtained with converged Monte Carlo simulations. Moreover, order two in space is verified for the different examples. 

In future work, we aim to adapt the previous FMM PDE of Proposition \ref{prop:FMMPDE} to allow the pricing of derivatives with payoffs including past fixings of RFRs. Many interest-rate contracts have payoffs depending not only on 
$R_k(T_k)=R(T_{k-1},T_k)$, but also on the backward-looking rates 
\begin{equation} \label{eq:generalR}
R(t,T) = \frac{1}{T-t} \left( e^{\int_{t}^{T} r(u)\rd u} - 1\right),
\end{equation}
for general $t<T$. Such derivatives with payment times and settings outside the FMM tenor structure $\{T_0,T_1,\ldots,T_N\}$ can not be priced directly with the FMM or its corresponding PDE.
In \cite{lyashenkoMercurio:Nov2019}, Lyashenko and Mercurio completed the FMM by embedding it into a Markovian HJM model with separable Cheyette volatility structure by aligning the HJM and FMM dynamics of the forward rates modeled by the FMM. Under this aligned FMM-Cheyette Markovian HJM model, it is possible to derive the dynamics of the short rate $r$. Besides, this short rate can be simulated by leveraging the realized paths of the simulated generalized FMM rates $R_k(t)$ and their volatilities $\nu_k(t)$. In order to adapt the FMM PDE to be able to price such path-dependent products, one possible approach is to introduce additional path-dependent variables. For example, to price a derivative whose payoff depends on the evolution of $R_k(\tau)$ for all times $\tau \in [T_{k-1},t]$, we can define a path-dependent variable of the general form
\begin{equation}\label{path_dep_proc}
I_k(t)=\int_{T_{k-1}}^{t} \zeta_k(\tau,R_k(\tau)) \rd \tau, \quad T_{k-1}<t<T_k.
\end{equation}
Here, the given function $\zeta_k$ vanishes for the time values outside the specified interval $[T_{k-1},t]$. The specification of the function $\zeta_k$ is part of future work and needs to be consistent with the FMM-HJM model in \cite{lyashenkoMercurio:Nov2019}. From now on, one must compute the SDE for $I_k$, and apply Itô's lemma to the derivative price function $\Pi$ depending on time, forward rates, and the additional state variable $I_k$. Finally, the PDE to price path-dependent RFR derivatives will be obtained by imposing no-arbitrage conditions. This strategy is classical and is explained for Asian options in equity markets in \cite{wilmottDewynneHowison}, for example. Discrete sampling is also possible within this path-dependent framework and can be implemented by forcing $\Pi$ to satisfy the so-called jump conditions.

Another possible extension is the consideration of sparse grid combination techniques for solving PDEs with higher spatial dimensions. The authors have already developed this strategy for pricing interest rate derivatives under the SABR/LIBOR model in \cite{LopezPerezVazquez:sisc}. This approach will also open the door to parallel computing.

\appendix
\section{Matrices in problem \eqref{ivp2d}}  \label{ap:matrixCoefs}
In this Appendix, we detail the coefficients of the matrices of the semi-discretized initial value problem \eqref{ivp2d} and how the linear systems involved in the one-stage AMFR-W1-method \eqref{amfrw1} are solved. 

In order to simplify this explanation, we denote $I^{(k)}$ as the identity matrix of dimension $M_k+1,\,k=1,\dots,N$. By using the tensor product notation ($A\otimes B=(a_{ij}B)$), we define the matrices
$$\begin{array}{l}
\tilde{{\cal A}}_1 = I^{(N)} \otimes I^{(N-1)} \otimes \dots \otimes I^{(2)} \otimes A_1,      \\
\tilde{{\cal A}}_k^{(l)} = I^{(N)} \otimes \dots \otimes I^{(k+1)} \otimes A_k^{(l)} \otimes I^{(k-1)}\otimes \dots \otimes I^{(1)},\,\,l=1,2,\,k=2,\dots,N.
\end{array}$$
The matrix $A_1=((A_1)_{ij})_{i,j=0}^M$ is a tridiagonal matrix of dimension $M_1+1$. The entries of its diagonals are
$$\begin{array}{l}
(A_1)_{j_1,j_1-1}=\left\{ \begin{array}{ll}
  \left( - \dfrac{1}{h_1} \dfrac{\tau_1}{1+\tau_1 x_{1,j_1}} +\dfrac{1}{h_1^2} \right)\, \dfrac{x_{1,j_1}^2}{2} , & \hbox{if } 1\le j_1\le M_1-1 ,\\[0.6pc]
   -\dfrac{1}{h_1}\dfrac{\tau_1}{1+\tau_1 x_{1,j_1}}  \, x_{1,j_1}^2 , & \hbox{if }  j_1= M_1, 
\end{array}
\right.      \\[1.2pc]
(A_1)_{j_1,j_1}=\left\{ \begin{array}{ll}
   - \dfrac{1}{h_1^2} \, x_{1,j_1}^2 , & \hbox{if } 0\le j_1\le M_1-1, \\[0.6pc]
   \dfrac{1}{h_1}\dfrac{\tau_1}{1+\tau_1 x_{1,j_1}}  \, x_{1,j_1}^2 , & \hbox{if }  j_1= M_1, 
\end{array}
\right.      \\[1.2pc]
(A_1)_{j_1,j_1+1}=
  \left(  \dfrac{1}{h_1} \dfrac{\tau_1}{1+\tau_1 x_{1,j_1}} +\dfrac{1}{h_1^2} \right)\, \dfrac{x_{1,j_1}^2}{2} , \quad  \hbox{if } 0\le j_1\le M_1-1.
       
\end{array}$$
For each $k=2,\dots,N,$ both matrices $A_k^{(1)}$ and $A_k^{(2)}$ are tridiagonal matrices of dimension $M_k+1$. Their respective diagonals' elements are
$$\begin{array}{ll}
(A_k^{(1)})_{j_k,j_k-1}=
  \left( \dfrac{\tau_k}{1+\tau_k x_{k,j_k}} \beta_{j_k,-1}^{(k)} +\dfrac{1}{2} \eta_{j_k,-1}^{(k)}
  \right)\, x_{k,j_k}^2 , & \hbox{if } 1\le j_k\le M_k ,     \\[1.2pc]
(A_k^{(1)})_{j_k,j_k}= \quad
  \left( \dfrac{\tau_k}{1+\tau_k x_{k,j_k}} \beta_{j_k,0}^{(k)} +\dfrac{1}{2} \eta_{j_k,0}^{(k)}
  \right)\, x_{k,j_k}^2 , & \hbox{if } 0\le j_k\le M_k  ,    \\[1.2pc]
(A_k^{(1)})_{j_k,j_k+1}=
  \left( \dfrac{\tau_k}{1+\tau_k x_{k,j_k}} \beta_{j_k,1}^{(k)} +\dfrac{1}{2} \eta_{j_k,1}^{(k)}
  \right)\, x_{k,j_k}^2 , & \hbox{if } 0\le j_k\le M_k -1   , 
\end{array}$$

$$\begin{array}{ll}
(A_k^{(2)})_{j_k,j_k-1}=
   \beta_{j_k,-1}^{(k)} \, x_{k,j_k} , & \hbox{if } 1\le j_k\le M_k   ,   \\[1.2pc]
(A_k^{(2)})_{j_k,j_k}= \beta_{j_k,0}^{(k)} \, x_{k,j_k} , & \hbox{if } 0\le j_k\le M_k  ,    \\[1.2pc]
(A_k^{(2)})_{j_k,j_k+1}=
  \beta_{j_k,1}^{(k)} \, x_{k,j_k} , & \hbox{if } 0\le j_k\le M_k -1    , \\[1.2pc]
\end{array}$$
where the finite difference coefficients $\beta_{j_k,\cdot}^{(k)}$ and $\eta_{j_k,\cdot}^{(k)}$ are given in \eqref{diff_coef}. 
Moreover, in order to compute the Jacobian ${\cal A}_k$ for $k=2,\dots,N$, the diagonal matrix ${\cal D}_k(t)=\hbox{diag}( (d_k(t))_J )_{J=0}^{M-1}$ of dimension $M$ is needed. Its entries are
\begin{equation}\label{dks}
    (d_k(t))_J=\ds\sum_{l=1}^{k-1}  \lambda_l(t)\, \rho_{kl} \,\dfrac{\tau_l}{1+\tau_l x_{l,j_l}} \, x_{l,j_l},\quad \hbox{with } (j_1,\dots,j_N)=\vartheta^{-1}(J).
\end{equation} 

Note that the first row of all of these matrices is null, since $x_{k,0}=0$, for all $k=1,\dots,N$.

As a consequence, the linear systems of the form 
$$(I-\nu\tau {\cal A}_1(t_n))K=R$$ of dimension $M$ in the AMFR-W1-method \eqref{amfrw1} are decoupled into $\prod_{k=2}^N (M_k+1)$ systems of dimension $M_1+1$ with coefficient matrix 
$(I-\nu\tau \lambda_1^2(t_n) A_1)$. 

\noindent Moreover, since the coefficients of the diagonal matrix ${\cal D}_k(t)$ defined in \eqref{dks} only depend on the indices $j_1,\dots,j_{k-1}$, the linear systems  in \eqref{amfrw1} of the form 
$$(I-\nu\tau {\cal A}_k(t_n))K=R, \quad \mbox{for } k=2,\dots,N,$$  
 can be decoupled into $\prod_{l\ne k}(M_l+1)$ linear systems of dimension $M_k+1$. 

\noindent More precisely, for each $k=2,\dots,N$, for each multi-index $(\dots,j_{k-1},j_{k+1},\dots)$ of $(N-1)$ integers with $0\le j_r\le M_r,\, r\ne k,$ the code computes
$$d_k(t_n)=\ds\sum_{l=1}^{k-1}  \lambda_l(t_n)\, \rho_{kl} \,\dfrac{\tau_l}{1+\tau_l x_{l,j_l}} \, x_{l,j_l}, $$
and solves a linear system of dimension $M_k+1$ with coefficient matrix $$(I-\nu \tau (\lambda_k^2(t_n)A_k^{(1)}+\lambda_k(t_n) d_k(t_n)A_k^{(2)})).$$

\section*{Acknowledgements}
S. Pérez-Rodríguez has been partially supported by the Spanish Ministry of Science and Innovation through project PID2022-141385NB-I00. J.G. López-Salas and C. Vázquez acknowledge the funding by the Spanish Mi\-nis\-try of Science and Innovation through the project PID2019-10858RB-I00 and PID2022-141058OB-I00, as well as from the Galician Government through grants ED431C 2022/47, both including FEDER financial support. Also J.G. López-Salas and C. Vázquez acknowledge the support received from the Centro de Investigaci\'on en Tecnolog\'{\i}as de la Informaci\'on y las Comunicaciones de Galicia, CITIC, funded by Xunta de Galicia and the European Union (European Regional Development Fund, Galicia 2014-2020 Program) through the grant ED431G 2019/01.

\bibliography{mybibfile}

\end{document}